\documentclass[10pt]{amsart}

\usepackage{amsfonts}
\usepackage{amssymb}
\usepackage{amsthm}
\usepackage{color}
\usepackage{latexsym}
\usepackage{bbm}
\usepackage{enumerate}
\usepackage{graphicx}
\usepackage{graphics}
\usepackage{float}
\usepackage{sidecap}
\usepackage{wrapfig}
\usepackage{subfig}

\usepackage{graphicx}
\usepackage{amsmath}
\usepackage{setspace}
\usepackage[tmargin=2.0cm, bmargin=2.0cm, lmargin=3.0cm, rmargin=3.0cm]{geometry}

\newcommand*\diff{\mathop{}\!\mathrm{d}}

\newtheorem{theorem}{Theorem}

\newtheorem{lem}{Lemma}[section]
\newtheorem{cor}[theorem]{Corollary}
\def\BEN{\begin{enumerate}}  \def\BI{\begin{itemize}}
\def\EEN{\end{enumerate}}   \def\EI{\end{itemize}}

\newtheorem{Lemma}{Lemma}
\newtheorem{Prop}{Proposition}

\title{Pricing insurance drawdown-type contracts \\with underlying L\'evy assets}

\author{Zbigniew Palmowski}
\address{Faculty of Pure and Applied Mathematics, Wroc\l aw University of Science and Technology, ul. Wyb. Wyspia\'nskiego 27, 50-370 Wroc\l aw, Poland}
\email{zbigniew.palmowski@gmail.com}

\author{Joanna Tumilewicz}
\address{Mathematical Institute, University of Wroc\l aw, pl. Grunwaldzki 2/4, 50-384 Wroc\l aw, Poland}
\email{joanna.tumilewicz@gmail.com}

\thanks{This work is partially supported by National Science Centre Grants No. 2015/17/B/ST1/01102
(2016-2019) and No. 2016/23/N/ST1/01189
(2017-2019).}

\date{\today}

\begin{document}

\begin{abstract}
In this paper we consider some insurance policies related to drawdown and drawup events
of log-returns for an underlying asset modeled by a spectrally negative geometric L\'evy process.
We consider four contracts, three of which were introduced in \cite{olympia} for a geometric Brownian motion.
The first one is an insurance contract where the protection buyer pays a constant premium
until the drawdown of fixed size of log-returns occurs.
In return he/she receives a certain insured amount at the drawdown epoch.
The next insurance contract provides protection from any specified drawdown with a drawup contingency.
This contract expires early if a certain fixed drawup event occurs prior to the fixed drawdown.
The last two contracts are extensions of the previous ones by an additional cancellation feature which allows
the investor to terminate the contract earlier.
We focus on two problems: calculating the fair premium $p$ for the basic contracts and identifying
the optimal stopping rule for the policies with the cancellation feature.
To do this we solve some two-sided exit problems related to  drawdown and drawup of spectrally negative L\'evy processes,
which is of independent mathematical interest.
We also heavily rely on the theory of optimal stopping.
\vspace{3mm}


\end{abstract}

\maketitle

\pagestyle{myheadings} \markboth{\sc Z.\ Palmowski
--- J.\ Tumilewicz} {\sc Pricing insurance L\'{e}vy-drawdown-type contracts}

\vspace{1.8cm}

\tableofcontents

\newpage

\section{Introduction}

The drawdown of a given process is the distance of
the current value away from the maximum value it has attained to date.
Similarly, the drawup is defined as the current
rise in the value over the running minimum.
Both have been customarily used as dynamic risk measures.
In fact,
the drawdown process not only provides dynamic measure of risk, but
it can also be viewed as giving measure of relative regret. Similarly the drawup process
can be viewed as providing measure of relative satisfaction.
Thus, a drawdown or
a drawup may signal the time when the investor may choose to change his/her investment position, which depends on
his/her perception of future moves of the market and his/her risk aversion.

The interest in the drawdown process has been strongly raised
by the recent financial crisis.
A large market drawdown may bring portfolio losses, liquidity shocks and even future recessions.
Therefore risk management of drawdown has become so important among practitioners; see
e.g. \cite{GZ} for portfolio optimization
under constraints on the drawdown process,
\cite{CZH, MA} for the distribution of the maximum drawdown of drifted Brownian
motion and the time-adjusted measure of
performance known as the Calmar ratio, and
\cite{PV, Vec1, Vec2} for the drawdown process as a dynamic measure of risk.
For an overview of the existing techniques for analysis of market crashes as well as
a collection of empirical studies of the drawdown process and the maximum
drawdown process, see Sornette \cite{Sorn}.

It is thus natural that fund managers have a strong incentive to seek insurance against
drawdown. In fact, as papers \cite{CZH, Vec1, Vec2} argue, some market-traded contracts, such as vanilla and look-back puts,
have only limited ability to insure against market drawdown. Therefore
the drawdown protection can be useful also for individual investors.

In this paper we follow Zhang et al. \cite{olympia} in pricing
some insurance contracts against drawdown (and drawup) of
log-returns of stock price modeled by an exponential L\'evy process
and identifying the optimal stopping rules.
We also identify for these contracts the so-called fair premium rates
for which the contract prices equal zero.

In its simplest form, the first drawdown
insurance contract involves a continuous premium payment by the investor
(protection buyer) to insure against a drawdown of log-returns of the underlying asset over a pre-specified level.
A possible buyer of this contract might think that a large drawdown is unlikely and he/she
might want to stop paying the premium.
Therefore we expand the simplest contract by adding a cancellation feature.
In this case, the investor receives the right to terminate the contract earlier and
in that case he/she pays a penalty for doing so.
We show that the investor's optimal cancellation time
is based on the first passage time of the drawdown of the log-return process.

Moreover, we also consider a related contract that protects
the investor from a drawdown of log-return of the asset price
preceding a drawup related to it.
In other words, the contract expires early if a drawup occurs
prior to a drawdown.
From the investor's perspective, when
a drawup occurs, there is little need to insure against a drawdown.
Therefore, this drawup contingency automatically stops the
premium payment, and it is an attractive feature that could potentially
reduce the cost of the drawdown insurance.
Finally, we also add a cancellation feature to this contract.

Zhang et al. \cite{olympia} only considered a risky asset modeled by the geometric Brownian motion.
However, in recent years, the empirical study of financial
data reveals that the distribution of the log-return of stock price exhibits
features which cannot be captured by the normal distribution such as heavy tails and
asymmetry.
With a view to replicating these features more effectively
and to reproducing a wide variety of implied volatility skews and smiles, there has
been a general shift in the literature to modeling a risky asset with an exponential L\'evy process rather than the exponential of a linear Brownian motion; see Kyprianou \cite{KIntr} and {\O}ksendal and Sulem \cite{OS} for overviews.
Therefore looking for a better fitting of the evolution of the stock
price process to real data, in this paper we price derivative securities in the
market by a general geometric spectrally negative
L\'evy process.
That is, the logarithm of the price of a risky asset in our case will be a process with stationary and independent increments
with no positive jumps.

The last contract analyzed in this paper taking into account drawdown and drawup with a cancellation feature is
considered for the first time in the literature. Although it is the most complex, it
produces very interesting and surprising results. In particular, we discover a new phenomenon for the
optimal stopping rule in this contract.
In the phenomenon, the investor's stopping rule is also at a first passage time of the drawdown of the log-return process,
similarly to the second contract without a drawup contingency. Still, the level of termination is different,
taking into account the drawup event.

Our approach is based on the classical fluctuation theory for  spectrally negative L\'evy processes
(related to so-called scale functions)
and some new exit identities for reflected L\'evy processes.
These new formulas identify two-sided exit problems for drawup and drawdown first passage times.
A key element of our approach is path analysis and the use of some results of
Mijatovi\'{c} and Pistorius \cite{Mijatovic1}.
We also heavily use optimal stopping theory.
In a market where the underlying dynamics for the stock price process is driven by the
exponential of a linear Brownian motion the valuation is transformed
into a free boundary problem.
However, if we allow jumps to appear in the sample paths of the dynamics of the stock price process,
this idea  breaks down.
To tackle these infinite
horizon problems we use the so-called ``guess and verify'' method.
For this method, one guesses
what the optimal value function and optimal stopping should be, and then tries to verify that candidate solution is
indeed the optimal one by testing it by means of a verification theorem.
This means that the value function identified by the guessed stopping rule
applied to the log-return price process
constructs a smallest, in some sense, discounted supermartingale.

In this paper we also analyze many particular examples and make an extensive numerical analysis
showing the dependence of the contract and stopping time on the model's parameters.
We mainly focus on the case when the logarithm of the asset price is a linear Brownian motion
or drift minus a compound Poisson process (a Cram\'er-Lundberg risk process).

The paper is organized as follows.
In Section \ref{sec:prel} we introduce the main definitions, notation, and the main fluctuation identities.
We analyze insurance contracts based on drawdown and additional drawup
in Sections \ref{sec:drawdown} and \ref{sec:drawup}, respectively. We finish by the numerical analysis in Section \ref{sec:examples} and Conclusions in Section \ref{sec:con}.

\section{Preliminaries}\label{sec:prel}
We work on a complete filtered probability space $(\Omega,\mathcal{F},\mathbb{P})$ satisfying the usual conditions. We model the logarithm of the underlying risky
asset price $\log S_t$ by
a spectrally negative L\'evy process $X_t$, that is, $S_t=\exp\{X_t\}$ is a geometric L\'evy process.
This means that $X_t$ is a stationary stochastic process with independent increments, right-continuous paths with left limits, and has only negative jumps.

Many identities will be given in terms of so-called scale functions which are defined in the following way.
We start by defining the so-called Laplace exponent of $X_t$:
\begin{equation}\label{psi}
\psi (\phi )=\log\mathbb{E}[e^{\phi X_1}],
\end{equation}
which is well defined for $\phi\geq 0$ due to the absence of positive jumps.
Recall that by the L\'evy-Khintchine theorem,
\begin{equation}\label{eq:exponent}
\psi(\phi)=\mu\phi
+\frac{1}{2}\sigma^{2}\phi^{2}+\int_{(0,\infty)}\big(\mathrm e^{-\phi
u}-1+\phi u\mathbbm{1}_{(u<1)}\big)\Pi(\diff u),
\end{equation}
which is analytic for $\mathfrak{Im}(\phi)\leq 0$, where $\mu$ and $\sigma\geq 0$ are real
and $\Pi$ is a so-called L\'evy measure such that $\int\left(1\wedge x^2\right)\Pi(\diff x)<\infty$.
It is easy to observe that $\psi$ is
zero at the origin, tends to infinity at infinity and is strictly
convex. We denote by $\Phi:[0,\infty)\rightarrow [0,\infty)$ the
right continuous inverse of $\psi$ so that
\begin{equation*}
\Phi(r)=\sup\{\phi>0:\psi(\phi)=r\} \quad \textrm{and} \quad
\psi(\Phi(r))=r \quad \text{for all} \ r \geq 0.
\end{equation*}
For $r\geq 0$ we define a continuous and strictly increasing function $W^{(r)}$ on $[0,\infty )$  with Laplace transform given by
\begin{align}
\int_0^\infty e^{-\phi u}W^{(r)}(u)\diff u=\frac{1}{\psi (\phi )-r}\quad \textrm{for}\quad \phi>\Phi(r),\label{Wq}
\end{align}
where $\psi$ is the Laplace exponent of $X_t$ given in \eqref{psi}.
$W^{(r)}$ is called the first scale function.
The second scale function is related to the first via:
\begin{align}
Z^{(r)}(u)=1+r\int_0^u W^{(r)}(\phi)\diff\phi. \label{Zq}
\end{align}
In this paper we will assume that
\begin{equation}\label{scalec1}
W^{(r)}\in \mathcal{C}^1(\mathbb{R}_{+})
\end{equation}
for $\mathbb{R}_{+}=[0,\infty)$.
This assumption is satisfied when the process $X_t$ has a non-trivial Gaussian component, or it is of unbounded variation, or the jumps have a density; see
\cite[Lem. 2.4]{kyprianoua}.
The scale functions are used in two-sided exit formulas:
\begin{align}
&\mathbb{E}_x\left[e^{-r\tau^+_a};\tau^+_a<\tau^-_0\right]=\frac{W^{(r)}(x)}{W^{(r)}(a)},\label{twosided1}\\
&\mathbb{E}_x \left[e^{-r\tau^-_0};\tau^-_0<\tau^+_a\right]=Z^{(r)}(x)-Z^{(r)}(a)\frac{W^{(r)}(x)}{W^{(r)}(a)},\label{twosided2}
\end{align}
where $x\leq a$, $r\geq 0$ and
\begin{align}
\tau^+_a=\inf\{t\geq 0 : X_t\geq a\},\qquad
\tau^-_a=\inf\{t\geq 0 : X_t\leq a\}
\end{align}
are the first passage times and we set $\inf\emptyset=\infty$. We  use the notation
$\mathbb{E}[\cdot\; \mathbbm{1}_{\{A\}}]=\mathbb{E}[\cdot; A]$ for any event $A$.

Set:
\begin{equation*}\label{supinf}
\overline{X}_t=\sup_{s\leq t} X_s,\qquad \underline{X}_t=\inf_{s\leq t} X_s.
\end{equation*}
In this paper, we analyze some insurance contracts related to the
drawdown and drawup processes of the log-return of the asset price $S_t$, that is,
to the drawdown and drawup processes of $X_t$.
The drawdown and drawup processes are Markov process and are defined as follows.
The drawdown is the difference between the running maximum of the underlying process and its current value and
the drawup is the difference between the current value and the running minimum.
Here, we additionally assume that the drawdown and drawup processes start from some points $y>0$ and $z>0$, respectively.
That is,
\begin{align}
D_t=\overline{X}_t\vee y-X_t,\qquad U_t=X_t-\underline{X}_t\wedge (-z).
\end{align}
Here $y$ and $-z$ can be interpreted as the historical maximum and the historical minimum of $X$.
Crucial for further work are the following first passage times of the drawdown process and the drawup process:
\begin{align}
\tau_D^+(a)=\inf\{t\geq 0 : D_t\geq a\},\qquad
\tau_D^-(a)=\inf\{t\geq 0 : D_t\leq a\},\label{tauDm}
\end{align}
\begin{align}
\tau_U^+(a)=\inf\{t\geq 0 : U_t\geq a\},\qquad
\tau_U^-(a)=\inf\{t\geq 0 : U_t\leq a\}.\label{tauU}
\end{align}
Later, we will use the following notational convention:
\begin{align}
\mathbb{P}_{|y}\left[\cdot\right]:=\mathbb{P}\left[\cdot|D_0=y\right],\nonumber\quad
\mathbb{P}_{|y|z}\left[\cdot\right]:=\mathbb{P}\left[\cdot|D_0=y,U_0=z\right],\nonumber\quad
\mathbb{P}_{x|y|z}\left[\cdot\right]:=\mathbb{P}\left[\cdot|X_0=x,D_0=y,U_0=z\right].\nonumber
\end{align}
Finally, we denote $\mathbb{P}_{x}\left[\cdot\right]:=\mathbb{P}\left[\cdot|X_0=x\right]$
with $\mathbb{P}=\mathbb{P}_{0}$, and $\mathbb{E}_{|y}, \mathbb{E}_{|y|z}, \mathbb{E}_{x|y|z}, \mathbb{E}_{x}, \mathbb{E}$
are the corresponding expectations.

We finish this section with two main formulas (the fist one is given in Mijatovi\'{c} and Pistorius \cite[Thm. 4]{Mijatovic1}
and the second follows from \eqref{twosided1})
that identify the joint laws of $\{\tau_U^+, \overline{X}_{\tau_U^+}$, $\underline{X}_{\tau^+_U}\}$
and $\{\tau_D^-(\theta)$, $\underline{X}_{\tau_D^-(\theta)}\}$:
\begin{align}
\mathbb{E}\left[e^{-r\tau^+_U(b)+u\underline{X}_{\tau^+_U(b)}};\overline{X}_{\tau^+_U(b)}<v\right]=&e^{ub}\frac{1+(r-\psi (u))\int_0^{b-v}e^{-uy}W^{(r)}(y)dy}{1+(r-\psi (u))\int_0^{b}e^{-uy}W^{(r)}(y)\diff y}\nonumber\\
&-e^{-u(b-v)}\frac{W^{(r)}(b-v)}{W^{(r)}(b)},\label{m2}
\end{align}
\begin{align}
\mathbb{E}_{|y}\left[e^{-r\tau_D^-(\theta)};\underline{X}_{\tau_D^-(\theta)}>-x\right]=\mathbb{E}_{x}\left[e^{-r\tau^+_{y-\theta+x}};\tau^+_{y-\theta+x}<\tau^-_0\right]=\frac{W^{(r)}(x)}{W^{(r)}(y-\theta+x)}.\label{distXD-}
\end{align}

\section{Drawdown insurance contract}\label{sec:drawdown}
\subsection{Fair premium}
In this section, we consider an insurance contract in which
the protection buyer pays a constant premium $p\geq 0$ continuously until a drawdown of log-returns of the asset price
of size $a>0$ occurs.
In return she/he receives the insured amount $\alpha\geq 0$ at the drawdown epoch.
Let $r\geq 0$ be the risk-free interest rate. The contract price is equal to the discounted value of the future cash-flows:
\begin{align}\label{f}
f(y,p) = \mathbb{E}_{|y}\left[ -\int_0^{\tau_D^+(a)}e^{-rt}p \diff t + \alpha e^{-r\tau_D^+(a)}\right].
\end{align}
Note that in this contract the investor wants to protect herself/himself from the asset price $S_t=e^{X_t}$ falling down from
the previous maximum by more than a fixed level $e^a$ for some $a>0$.
In other words, she/he believes that even if the price goes up again after the first drawdown of size $e^a$,
this will not bring her/him sufficient profit.
Therefore, she/he is ready to enter into this type of contract to reduce loss by getting $\alpha>0$
at the drawdown epoch.

Note that
\begin{equation}\label{dodane1}
f(y,p)= \left ( \frac{p}{r}+\alpha \right ) \xi (y) -\frac{p}{r},
\end{equation}
where
\begin{equation}\label{xi}
\xi(y) := \mathbb{E}_{|y}\left[ e^{-r\tau_D^+(a)}\right]
\end{equation}
is the conditional Laplace transform of $\tau_D^+ (a)$ given that $D_0=y\in (0,a)$.
To price the contract \eqref{f} we start by identifying
the crucial function $\xi$.

\begin{Prop}\label{Ksi}
The conditional Laplace transform $\xi (\cdot)$ is given by
\begin{align}\label{ksi}
\xi (y) = Z^{(r)}(a-y)-rW^{(r)}(a-y)\frac{W^{(r)}(a)}{W^{\prime (r)}(a)}.
\end{align}
\end{Prop}
\begin{proof}
Note that $\tau_D^-(0)$ is the first time that the drawdown process $D_t$ passes level $0$,
which means that the process $X_t$ attains its historical maximum.
This occurs in a continuous way by assumed spectral negativity of the L\'evy process $X$.
By the strong Markov property of $D_t$ at $\tau_D^-(0)$ we have
\begin{align}
\xi (y)=& \mathbb{E}_{|y}\left[ e^{-r\tau_D^+(a)}\right] \nonumber \\
=& \mathbb{E}_{|y}\left[ e^{-r\tau_D^+(a)};\tau_D^+(a)<\tau_D^-(0)\right] + \mathbb{E}_{|y}\left[ e^{-r\tau_D^-(0)};\tau_D^-(0)<\tau_D^+(a)\right]\xi (0) \nonumber \\
=&\mathbb{E}_{a-y}\left[ e^{-r\tau^-_0};\tau^-_0<\tau^+_a\right]+ \mathbb{E}_{a-y}\left[ e^{-r\tau^+_a};\tau^+_a<\tau^-_0\right]\xi (0)\nonumber\\
=& Z^{(r)}(a-y) - Z^{(r)}(a)\frac{W^{(r)}(a-y)}{W^{(r)}(a)}+\frac{W^{(r)}(a-y)}{W^{(r)}(a)}\xi (0), \label{ksiy}
\end{align}
where the third equality follows from the two-sided exit formulas \eqref{twosided1} - \eqref{twosided2}.
Therefore the problem of finding $\xi$ is reduced to identifying $\xi (0)$.
The latter can be obtained from \cite[Prop. 2(ii), p. 191]{Pistorius}:
\begin{align}\label{ksi0}
\xi (0) =  Z^{(r)}(a) - rW^{(r)}(a)\frac{W^{(r)}(a)}{W^{\prime (r)}(a)}.
\end{align}
This completes the proof.
\end{proof}
Thus we have the following theorem.
\begin{theorem}
The value of the contract \eqref{f} is given in \eqref{dodane1} for $\xi$ given in \eqref{ksi}.
\end{theorem}

It is a fair situation for both sides, the insurance company and the investor, when the contract price at the conclusion time equals $0$.
Therefore, we say that the premium $p^*$ is {\it fair} when
$$f(y,p^*)=0.$$
From \eqref{dodane1} using Proposition \ref{Ksi} we derive the following theorem.
\begin{theorem}
For the contract \eqref{f} the fair premium equals
\begin{equation}
p^* = \frac{r\alpha\xi (y)}{1-\xi (y)}.\label{p*}
\end{equation}
\end{theorem}

\subsection{Cancellation feature}
We now extend the previous contract by adding a possibility of cancellation.
In other words, we give the investor the right to terminate the contract by paying a fixed fee $c\geq 0$ at
any time prior to a pre-specified drawdown of log-return of the asset price of size $a>0$.
This contract is addressed to those investors who are not willing to pay the premium any longer after they stop
to believe that a large drawdown of the asset price may happen.
The contract value then equals
\begin{align}
F(y,p) = \sup\limits_{\tau\in\mathcal{T}}\mathbb{E}_{|y}\Bigg[-\int_0^{\tau_D^+(a)\wedge\tau}e^{-rt}p \diff t - ce^{-r\tau}\mathbbm{1}_{(\tau<\tau_D^+(a))}+ \alpha e^{-r\tau_D^+(a)}\mathbbm{1}_{(\tau_D^+(a)\leq\tau)}\Bigg],\label{F}
\end{align}
where $\mathcal{T}$ is the family of all $\mathcal{F}_t$-stopping times.

One of the main goals of this paper is to identify the optimal stopping rule $\tau^*$ that realizes the price $F(y,p)$.
We start from a simple observation.
\begin{Prop}\label{Prop_decomp}
The cancellable drawdown insurance value admits the following decomposition:
\begin{align}\label{Fdecomposition}
F(y,p)=f(y,p)+G(y,p),
\end{align}
where
\begin{align}
&G(y,p):=\sup\limits_{\tau\in\mathcal{T}}g_\tau(y,p),\label{gtau}\\
&g_\tau(y,p):=\mathbb{E}_{|y}\left[e^{-r\tau}\tilde{f}(D_{\tau},p); \tau<\tau_D^+(a)\right],\label{gtaub}\\
&\tilde{f}(y,p):=-f(y,p)-c\label{numer}
\end{align}
and $f(\cdot,\cdot)$ is defined in (\ref{f}).
\end{Prop}
\begin{proof}
Using $\mathbbm{1}_{(\tau\geq\tau_D^+(a))}=1-\mathbbm{1}_{(\tau <\tau_D^+(a))}$ in (\ref{F}) we obtain
\begin{eqnarray}
\lefteqn{F(y,p)=\mathbb{E}_{|y}\left[ -\int_0^{\tau_D^+(a)}e^{-rt}p \diff t +\alpha e^{-r\tau_D^+(a)}\right]} \nonumber \\ &&+\sup\limits_{\tau\in\mathcal{T}}\mathbb{E}_{|y}\left[\int_{\tau\wedge\tau_D^+(a)}^{\tau_D^+(a)}e^{-rt}p \diff t - \alpha e^{-r\tau_D^+(a)}\mathbbm{1}_{\left(\tau<\tau_D^+(a)\right)} - ce^{-r\tau}\mathbbm{1}_{\left(\tau <\tau_D^+(a)\right)}\right]. \nonumber
\end{eqnarray}
Note that the first summand does not depend on $\tau$. The second depends on $\tau$ only through $\tau <\tau_D^+(a)$. Then by the strong Markov property,
\begin{eqnarray}
\lefteqn{F(y,p)=f(y,p)\nonumber}\\
&&\qquad+\sup\limits_{\tau\in\mathcal{T}}\mathbb{E}_{|y}\left[\int_{\tau}^{\tau_D^+(a)}e^{-rt}p \diff t - \alpha e^{-r\tau_D^+(a)}\mathbbm{1}_{\left(\tau<\tau_D^+(a)\right)} - ce^{-r\tau}\mathbbm{1}_{\left(\tau <\tau_D^+(a)\right)}\right] \nonumber \\
&&=f(y,p)+\sup\limits_{\tau\in\mathcal{T}}\mathbb{E}_{|y}\left[e^{-r\tau}\mathbb{E}_{|D_{\tau}}\left(\int_{0}^{\tau_D^+(a)}e^{-rt}p\diff t - \alpha e^{-r\tau_D^+(a)} - c\right);\ \tau <\tau_D^+(a)\right]. \nonumber
\end{eqnarray}
This completes the proof.
\end{proof}

Observe now that $\tilde{f}(y,p)$ in (\ref{numer}) is a decreasing function of $y$.
Thus, if $\tilde{f}(0+,p)<0$, then the optimal stopping strategy for the investor is to never terminate the contract, that is, $\tau = \infty$.
To eliminate this trivial case we assume from now on that
\begin{equation}\label{mainzalozenia}
\tilde{f}(0+,p)>0,
\end{equation}
which is equivalent to
\begin{eqnarray}
\frac{p}{r}-c>\left(\frac{p}{r}+\alpha\right)\xi(0+)\geq 0.\label{war1}
\end{eqnarray}

In order to determine the optimal cancellation strategy for our contract it is sufficient to solve the optimal stopping problem represented by the second summand in (\ref{Fdecomposition}), that is, to identify $G(y,p)$. We will use the ``guess and verify'' approach. This means that we first guess a candidate stopping rule and then verify its optimality
using the Verification Lemma below.
\begin{lem}\label{war}
Let $\Upsilon_t$ be a right-continuous process living in some Borel state space $\mathbb{B}$ and
killed at some $\mathcal{F}^\Upsilon_t$-stopping time $\tau_0$, where $\mathcal{F}^\Upsilon_t$ is a right-continuous natural filtration of $\Upsilon$.
Consider the following stopping problem:
\begin{equation}\label{BVP}
v(\phi)=\sup\limits_{\tau\in\mathcal{T}^\Upsilon}\mathbb{E}\left[e^{-r\tau}V(\Upsilon_\tau)|\Upsilon_0=\phi\right]
\end{equation}
for some function $V$ and the family $\mathcal{T}^\Upsilon$ of $\mathcal{F}^\Upsilon_t$-stopping times.
Assume that
\begin{equation}\label{verlemass}
\mathbb{P}(\lim_{t\rightarrow \infty}e^{-rt}V(\Upsilon_t)<\infty|\Upsilon_0=\phi)=1.
\end{equation}
The pair $(v^*,\tau^*)$ is a solution of the stopping problem \eqref{BVP}, that is,
$$v^*(\phi):=\mathbb{E}\left[e^{-r\tau^*}V(\Upsilon_{\tau^*})|\Upsilon_0=\phi\right],$$
if the following conditions are satisfied:
\begin{enumerate}[(i)]
\item $v^*(\phi)\geq V(\phi)$ for all $\phi\in\mathbb{B}$, \label{i}
\item the process $e^{-rt}v^*(\Upsilon_t)$ is a right continuous supermartingale. \label{ii}
\end{enumerate}
\end{lem}
\begin{proof}
The proof follows the same arguments as the proof of \cite[Lem. 9.1, p. 240]{KIntr}; see also \cite[Th. 2.2, p. 29]{Peskir}.
\end{proof}

Using Verification Lemma \ref{war} we will prove that the first passage time of
the drawdown process $D_t$ below some level $\theta$ is the optimal stopping time for \eqref{gtau}
(hence also for (\ref{Fdecomposition})). That is, we will prove that
\begin{equation}\label{tau*pierwsze}
\tau^*=\tau_D^-(\theta)\in\mathcal{T}
\end{equation}
for an appropriate $\theta\in [0,a)$.

For the stopping rule (\ref{tau*pierwsze}) and for $y>\theta$ we will explicitly compute
$g_{\tau_D^-(\theta)}(y,p)$ given in \eqref{gtaub}.
Note that if $y>\theta$ then
\begin{align}
g(y,p,\theta):=g_{\tau_D^-(\theta)}(y,p)=\tilde{f}(\theta,p)\mathbb{E}_{|y}\left[e^{-r\tau_D^-(\theta)}; \tau_D^-(\theta)<\tau_D^+(a)\right]=\tilde{f}(\theta,p)\frac{W^{(r)}(a-y)}{W^{(r)}(a-\theta)}.\label{g}
\end{align}
Furthermore, if $y\leq\theta$ then the investor will terminate the contract immediately:
\begin{eqnarray}
g(y,p,\theta):=\mathbb{E}_{|y}\left[e^{-r\tau_D^-(\theta)}\tilde{f}(D_{\tau_D^-(\theta)},p);\tau_D^-(\theta)<\tau_D^+(a)\right]=\tilde{f}(y,p).\label{g=f}
\end{eqnarray}
Thus, for $\theta\in [0,y]$ we have
\begin{align}
F(y,p,\theta)&=f(y,p)+g(y,p,\theta) \nonumber\\
&=\left(\frac{p}{r}+\alpha\right)Z^{(r)}(a-y)+\left(\frac{p}{r}-c\right)\frac{W^{(r)}(a-y)}{W^{(r)}(a-\theta)}-\left(\frac{p}{r}+\alpha\right)Z^{(r)}(a-\theta)\frac{W^{(r)}(a-y)}{W^{(r)}(a-\theta)}-\frac{p}{r}.
\end{align}
Recall that, from Proposition \ref{Prop_decomp}, the value function $F(y,p)$ depends on the stopping time only through the function $G(y,p)$. If a chosen stopping time $\tau_D^-(\theta)$ is a correct guess, then the optimal level $\theta^*$ has to maximize the function $G(y,p)=\sup\limits_{\theta\geq 0}g(y,p,\theta)=g(y,p,\theta^*)$. Thus, from \eqref{Fdecomposition}, the same $\theta^*$ maximizes the value function $F(y,p)$.
We define
\begin{align}\label{teta*}
\theta^*:=\inf\left\{ \theta\in [0,a):\ \frac{\partial}{\partial\theta}g(y,p,\theta)=0\quad\text{and }\quad  g(y,p,\varsigma)\leq g(y,p,\theta)\ \textrm{for all $\varsigma\geq 0$}\right\}.
\end{align}
Note that $\theta^*>0$ because $g(y,p,\theta)$ increases at $\theta=0$. Indeed,
\begin{align}
\left[\frac{\partial}{\partial\theta}g(y,p,\theta)\right]_{|_{\theta=0}}=\tilde{f}^{\prime}(0)\frac{W^{(r)}(a-y)}{W^{(r)}(a)}+\tilde{f}(0)\frac{W^{(r)}(a-y)W^{\prime (r)}(a)}{(W^{(r)}(a))^2}>0,\nonumber
\end{align}
where the inequality follows from assumption \eqref{mainzalozenia} and the fact that $\tilde{f}^{\prime}(0)=-(\frac{p}{r}+\alpha)\xi^{\prime}(0)=0$.
We will now verify that \eqref{tau*pierwsze} indeed holds true, that is, $\tau_D^-(\theta^*)$ is an optimal stopping rule.
\begin{theorem}\label{th1}
Assume that \eqref{war1} holds.
The stopping time $\tau_D^-(\theta^*)$, with $\theta^*$ defined in \eqref{teta*},
is the optimal stopping rule for the stopping problems \eqref{gtau} and \eqref{F}.
Moreover, the price of the drawdown insurance contract with the cancellation
feature equals $F(y,p)=f(y,p)+g(y,p,\theta^*)$.
\end{theorem}
\begin{proof}
Based on the optimal stopping problem \eqref{gtau} it is sufficient to check that
$\tau^*=\tau_D^-(\theta^*)$ fulfills the two conditions of Verification Lemma \ref{war}
with $\Upsilon_t=D_t$, $\mathbb{B}=\mathbb{R}_+$, $\tau_0=\tau_D^+(a)$, $V(x)=\tilde{f}(x,p)$.
Note that in this case the assumption \eqref{verlemass} is clearly satisfied.
In order to prove (\ref{i}) of Verification Lemma \ref{war} it suffices to show that
$g(y,p,\theta)-\tilde{f}(y,p)\geq 0$ for some $\theta$.
Observe that taking $\theta=y$ for $y\in (0,a)$ produces
\begin{align}
&g(y,p,y)-\tilde{f}(y,p)=\mathbb{E}_{|y}\left[e^{-r\tau_D^-(y)}\tilde{f}(D_{\tau_D^-(y)},p);\tau_D^-(y)<\tau_D^+(a)\right]-\tilde{f}(y,p)=0,
\end{align}
where $\tilde{f}(y,p)$ is given in \eqref{numer}.
Thus (\ref{i}) follows from the fact that
$\theta^*$ maximizes $g(y,p,\cdot)$.

Now we will prove (\ref{ii}). Note that
the stopped process $$e^{-r(t\wedge\tau_D^+(a)\wedge\tau_D^-(\theta^*))}g(D_{t\wedge\tau_D^+(a)\wedge\tau_D^-(\theta^*)},p,\theta^*)$$
is a martingale. This follows from the strong Markov property: from the definition of $g(y,p, \theta^*)$ in \eqref{g}
and \eqref{g=f}
we have
\begin{align}
&\mathbb{E}\left[e^{-r(\tau_D^+(a)\wedge\tau_D^-(\theta^*))}g(D_{\tau_D^+(a)\wedge\tau_D^-(\theta^*)},p,\theta^*)|\ \mathcal{F}_{t\wedge\tau_D^+(a)\wedge\tau_D^-(\theta^*)}\right]\nonumber\\
&\qquad=\mathbb{E}\left[e^{-r(\tau_D^+(a)\wedge\tau_D^-(\theta^*))}\mathbb{E}_{|D_{\tau_D^+(a)\wedge\tau_D^-(\theta^*)}}\left[e^{-r\tau_D^-(\theta^*)}\tilde{f}(\theta^*,p);\tau_D^-(\theta^*)<\tau_D^+(a)\right]\big|\ \mathcal{F}_{t\wedge\tau_D^+(a)\wedge\tau_D^-(\theta^*)}\right]\nonumber\\
&\qquad=\mathbb{E}\left[e^{-r(t\wedge\tau_D^+(a)\wedge\tau_D^-(\theta^*))}\mathbb{E}_{|D_{t\wedge \tau_D^+(a)\wedge\tau_D^-(\theta^*)}}\left[e^{-r\tau_D^-(\theta^*)}\tilde{f}(\theta^*,p);\tau_D^-(\theta^*)<\tau_D^+(a)\right]\big|\ \mathcal{F}_{t\wedge\tau_D^+(a)\wedge\tau_D^-(\theta^*)}\right]\nonumber\\
&\qquad=e^{-r(t\wedge\tau_D^+(a)\wedge\tau_D^-(\theta^*))}g(D_{t\wedge\tau_D^+(a)\wedge\tau_D^-(\theta^*)},p,\theta^*).\nonumber
\end{align}
Hence,
\[\mathcal{A}_Dg(y, p,\theta^*)-rg(y, p,\theta^*)=0\]
for all $y\in (\theta^*, a)$ and for the full generator $\mathcal{A}_D$ of the process $D$.
Moreover, from (\ref{g=f}) we know that for $y\in (0,\theta^*)$ we have $g(y,p,\theta^*)=\tilde{f}(y,p)$.
Therefore, for $y\in (0,\theta^*)$,
\begin{align}
\mathcal{A}_Dg(y,p,\theta^*)-rg(y,p,\theta^*)&=\mathcal{A}_D\tilde{f}(y,p)-r\tilde{f}(y,p)\nonumber\\
&=-\mathcal{A}_Df(y,p)+rf(y,p)+rc\nonumber\\
&=-\left(\frac{p}{r}+\alpha\right)\left[\mathcal{A}_D\xi(y)-r\xi(y)\right]-r\left(\frac{p}{r}-c\right).\nonumber
\end{align}
Now, the strong Markov property of the process $D_t$ implies that
the process $e^{-r(t\wedge \tau_D^+(a))}\xi(D_{t\wedge \tau_D^+(a)})=\mathbb{E}_{|y}[e^{-rs}\xi(D_s)|\mathcal{F}_{t\wedge \tau_D^+(a)}]$
is a martingale. Hence $\mathcal{A}_D\xi(y)-r\xi(y)=0$ for $y\in (0,\theta^*)$ since $\theta^*< a$.

Thus the process $e^{-r(t\wedge\tau_D^+(a))}g(D_{t\wedge\tau_D^+(a)},p,\theta^*)$ is a supermartingale
because for $y\in(0,\theta^*)$ we have
\begin{align}
\mathcal{A}_Dg(y,p,\theta^*)-rg(y,p,\theta^*)=-r\left(\frac{p}{r}-c\right)\leq 0,\nonumber
\end{align}
where the last inequality follows from the assumption (\ref{war1}).
This completes the proof.
\end{proof}

\section{Incorporating a drawup contingency}\label{sec:drawup}
\subsection{Fair premium}
The investor might like to buy a contract which has some maturity conditions,
meaning that a contract will end when these conditions are fulfilled.
Therefore in this paper we also consider
an insurance contract which provides protection against any specified drawdown of log-return of the asset price
with a certain drawup contingency.
In particular, this contract may expire earlier if a fixed drawup event occurs prior to some fixed drawdown.
Choosing the drawup event is very natural since it corresponds to some market upward trend,
and therefore the investor may want to stop paying the premium when this event happens.
Under a risk-neutral measure the value of this contract equals
\begin{align}
k(y,z,p):=&\mathbb{E}_{|y|z}\left[-\int_0^{\tau^+_D (a)\wedge\tau^+_U(b)}e^{-rt}p\diff t+\alpha e^{-r\tau^+_D(a)}\mathbbm{1}_{\left(\tau^+_D(a)\leq\tau^+_U(b)\right)}\right]\label{kyzp}\end{align}
for some fixed $a>b>0$.
First we will compute this value function and then we will identify the fair premium $p^*$ under which
\begin{equation}\label{fairk}
k(y,z,p^*)=0.
\end{equation}
Note that
\begin{align}
k(y,z,p)=&\left(\frac{p}{r}+\alpha\right)\nu(y,z)+\frac{p}{r}\lambda(y,z)-\frac{p}{r},\label{fdrawup}
\end{align}
where
\begin{align}
\nu(y,z)&:=\mathbb{E}_{|y|z}\left[e^{-r\tau^+_D(a)};\ \tau_D^+(a)\leq\tau_U^+(b)\right],\nonumber\\
\lambda(y,z)&:=\mathbb{E}_{|y|z}\left[e^{-r\tau_U^+(b)};\ \tau^+_U(b)<\tau^+_D(a)\right].\nonumber
\end{align}

To get formulas for $\nu$ and $\lambda$ we have to make some additional observations.
\begin{Prop}\label{propnowe} Let $y$ and $z$ denote the starting positions for the drawdown and drawup processes, respectively.
For $a>b\geq 0$ the following events are equivalent:
\begin{align}
&\left\{\tau_U^+(b)<\tau_D^+(a),\ D_0=y,\ U_0=z\right\}=\{\tau^+_{b-z}<\tau^-_{(y-a)\vee (-z)}\}\nonumber\\
&\qquad\cup\left\{\overline{X}_{\tau_U^+(b)}\vee y-\underline{X}_{\tau_U^+(b)}<a,\ \underline{X}_{\tau_U^+(b)}\leq-z\right\},\label{pierwsze}\\
&\left\{\tau_D^+(a)<\tau_U^+(b),\ D_0=y,\ U_0=z\right\}=\{\tau^-_{y-a}<\tau^+_{(b-z)},\  y-a\geq -z\}\nonumber\\
&\qquad\cup\left\{\overline{X}_{\tau_U^+(b)}\vee y-\underline{X}_{\tau_U^+(b)}\geq a,\ \underline{X}_{\tau_U^+(b)}\leq -z,\ \overline{X}_{\tau_U^+(b)}\leq b-z,\ y-a<-z\right\}.\label{drugie}
\end{align}
\end{Prop}
\begin{proof}
We use geometric path arguments.
To prove \eqref{pierwsze} note that
\begin{align}
\left\{\tau_U^+(b)<\tau_D^+(a),\ D_0=y,\ U_0=z\right\}=&\left\{\tau_U^+(b)<\tau_D^+(a),\ \underline{X}_{\tau_U^+(b)}>-z,\ D_0=y,\ U_0=z\right\}\nonumber\\
&\cup\left\{\tau_U^+(b)<\tau_D^+(a),\ \underline{X}_{\tau_U^+(b)}\leq-z,\ D_0=y\right\}.\nonumber
\end{align}
The event $\{\underline{X}_{\tau_U^+(b)}>-z,\ U_0=z\}$ is equivalent to the requirement that $X_{\tau_U^+(b)}=b-z$ and
that $X_t$ cannot cross the $y-a$ level before $\tau_U^+(b)$.
Therefore,
\begin{align}
\left\{\tau_U^+(b)<\tau_D^+(a),\ \underline{X}_{\tau_U^+(b)}>-z,\ D_0=y,\ U_0=z\right\}=\left\{\tau^+_{b-z}<\tau^-_{(y-a)\vee (-z)},\ D_0=y,\ U_0=z\right\}.\label{pierobs}
\end{align}
Now, let us consider the case when $\underline{X}_{\tau_U^+(b)}\leq -z$. If $X_{\tau_U^+(b)}<\overline{X}_{\tau_U^+(b)}\vee y$ then $\overline{D}_{\tau_U^+(b)}=\overline{X}_{\tau_U^+(b)}\vee y - X_{\tau_U^+(b)}$ and it has to be less than $a$, since the drawup happens before the drawdown. On the other hand, if $X_{\tau_U^+(b)}=\overline{X}_{\tau_U^+(b)}\vee y$ then $b=U_{\tau_U^+(b)}=\overline{X}_{\tau_U^+(b)}\vee y-\underline{X}_{\tau_U^+(b)}<a$, since $b<a$.
Thus, we get
\begin{align}
&\left\{\tau_U^+(b)<\tau_D^+(a),\ \underline{X}_{\tau_U^+(b)}\leq -z,\ D_0=y,\ U_0=z\right\}\nonumber\\
&\qquad=\left\{
\overline{X}_{\tau_U^+(b)}\vee y-\underline{X}_{\tau_U^+(b)}<a,\ \underline{X}_{\tau_U^+(b)}\leq-z,\ D_0=y,\ U_0=z\right\}.\label{drobs}
\end{align}
Observations \eqref{pierobs} and \eqref{drobs} complete the proof of \eqref{pierwsze}.

To prove \eqref{drugie} we again consider two scenarios:
\begin{align}
\left\{\tau_D^+(a)<\tau_U^+(b),\ D_0=y,\ U_0=z\right\}=&\left\{\tau_D^+(a)<\tau_U^+(b),\ D_0=y,\ U_0=z,\ y-a<-z\right\}\nonumber\\
&\cup\left\{\tau_D^+(a)<\tau_U^+(b),\ D_0=y,\ U_0=z,\ y-a\geq -z\right\}.\nonumber
\end{align}
The case $y-a>-z$ together with the assumption $b\leq a$ implies that $b-a<y$.
This means that the event when a drawdown occurs before a drawup is the same as the one when
the process $X$ crosses $y-a$ before it hits $b-z$. That is,
\begin{align}
\left\{\tau_D^+(a)<\tau_U^+(b),\ D_0=y,\ U_0=z,\ y-a\geq-z\right\}=\left\{\tau^-_{y-a}<\tau^+_{b-z},\ D_0=y,\ U_0=z,\ y-a\geq-z\right\}.\nonumber
\end{align}
If $y-a\leq -z$ then $X$ crosses level $-z$ before the drawdown event occurs. Additionally, $X$ can cross level $y$ but it cannot cross level $b-z$, because otherwise a drawup would occur. Thus,
\begin{align}
&\left\{\tau_D^+(a)<\tau_U^+(b),\ D_0=y,\ U_0=z,\ y-a<-z\right\}\nonumber\\
&\qquad=\left\{\overline{X}_{\tau_U^+(b)}\vee y-\underline{X}_{\tau_U^+(b)}>a,\ \underline{X}_{\tau_U^+(b)}\leq-z,\ \underline{X}_{\tau_U^+(b)}\leq b-z,\ y-a<-z\right\}.\nonumber
\end{align}
This completes the proof of \eqref{drugie}.
\end{proof}
Note that, for $b<a$, we have
\begin{align}
\mathbb{E}_{|y|z}\left[e^{-r\tau_U^+(b)};\ \tau_D^+(a)<\tau_U^+(b)\right]=\mathbb{E}_{|y|z}\left[e^{-r\tau_D^+(a)};\ \tau_D^+(a)<\tau_U^+(b)\right]\mathbb{E}\left[e^{-r\tau_U^+(b)}\right]\nonumber
\end{align}
Proposition \ref{propnowe} and the above observation yield the following crucial corollary.
\begin{cor}\label{lambdanu}
For $a>b$ we have
\begin{align}
&\nu (y,z)=\mathbb{E}\left[e^{-r\tau^-_{y-a}};\tau^-_{y-a}<\tau^+_{(b-z)}\right]\mathbbm{1}_{(y+z\geq a)}\nonumber\\
&\qquad+\mathbb{E}\left[e^{-r\tau_U^+(b)};\overline{X}_{\tau_U^+(b)}\vee y-\underline{X}_{\tau_U^+(b)}\geq a,\underline{X}_{\tau_U^+(b)}\leq -z,\ \overline{X}_{\tau_U^+(b)}\leq b-z\right]\frac{\mathbbm{1}_{(y+z<a)}}{\mathbb{E}\left[e^{-r\tau_U^+(b)}\right]}\nonumber
\end{align}
and
\begin{align}
&\lambda (y,z)=\mathbb{E}\left[e^{-r\tau^+_{b-z}};\tau^+_{b-z}<\tau^-_{(y-a)\vee (-z)}\right]\nonumber\\
&\qquad+\mathbb{E}\left[e^{-r\tau_U^+(b)};\overline{X}_{\tau_U^+(b)}\vee y-\underline{X}_{\tau_U^+(b)}<a,\quad \underline{X}_{\tau_U^+(b)}\leq -z\right].\nonumber
\end{align}
\end{cor}
Both functions $\lambda$ and $\nu$ can now be calculated by taking the inverse Laplace transform of (\ref{m2}).

\begin{theorem}\label{thm2}
The price of the contract \eqref{kyzp} is given in \eqref{fdrawup} with $\lambda$ and $\nu$ identified in
Corollary \ref{lambdanu}.
\end{theorem}
From \eqref{fdrawup} we obtain the following theorem.
\begin{theorem}
For the contract \eqref{kyzp} the fair premium defined in \eqref{fairk} equals
\begin{equation}
p^* = \frac{r\alpha\nu(y,z)}{1-\lambda(y,z)-\nu(y,z)},\label{p*drawup}
\end{equation}
where the functions $\lambda$ and $\nu$ are given in Corollary \ref{lambdanu}.
\end{theorem}

\subsection{Cancellation feature}
We will also consider the possibility of terminating the previous contract.
Now, the protection buyer can terminate the position by paying a fee $c\geq 0$. The value of the contract then equals
\begin{align}
K(y,z,p):=&\sup\limits_{\tau\in\mathcal{T}}\mathbb{E}_{|y|z}\Bigg[-\int_0^{\tau^+_D (a)\wedge\tau^+_U(b)\wedge\tau}e^{-rt}p\diff t\nonumber\\
&+\alpha e^{-r\tau^+_D(a)}\mathbbm{1}_{\left(\tau^+_D(a)<\tau^+_U(b)\wedge\tau\right)}-ce^{-r\tau}\mathbbm{1}_{\left(\tau<\tau^+_D(a)\wedge\tau^+_U(b)\right)}\Bigg].\label{K}
\end{align}
As in the case of a cancellable drawdown contract, we can represent the contract value function as the sum of two parts:
one without cancellation feature and one that depends on the stopping time $\tau$.
\begin{Prop}\label{Kdecompose}
The cancellable drawup insurance value admits the decomposition
\begin{align}\label{kdecompose}
K(y,z,p)=k(y,z,p)+H(y,z,p),
\end{align}
where
\begin{align}
&H(y,z,p):=\sup\limits_{\tau\in\mathcal{T}}h_{\tau}(y,z,p),\label{fh}\\
&h_{\tau}(y,z,p):=\mathbb{E}_{|y|z}\left[e^{-r\tau}\tilde{k}(D_{\tau},U_{\tau},p);\ \tau<\tau_D^+(a)\wedge\tau_U^+(b)\right],\label{htau}\\
&\tilde{k}(y,z,p):=-k(y,z,p)-c\label{ktilde}
\end{align}
and $k(\cdot,\cdot,\cdot)$ is given in \eqref{fdrawup}.
\end{Prop}
\begin{proof} Using $\mathbbm{1}_{(\tau_D^+(a)<\tau_U^+(b)\wedge\tau)}=\mathbbm{1}_{(\tau_D^+(a)<\tau_U^+(b))}-\mathbbm{1}_{(\tau<\tau_D^+(a)<\tau_U^+(b))}$ we obtain
\begin{align}
&K(y,z,p)=\mathbb{E}_{|y|z}\left[-\int_0^{\tau_D^+ (a)\wedge\tau_U^+(b)}e^{-rt}p\diff t+\alpha e^{-r\tau_D^+(a)}\mathbbm{1}_{\left(\tau_D^+(a)<\tau_U^+(b)\right)}\right]\nonumber\\
&\qquad +\sup\limits_{\tau\in\mathcal{T}}\mathbb{E}_{|y|z}\Bigg[\int_{\tau\wedge\tau_D^+(a)\wedge\tau_U^+(b)}^{\tau_D^+(a)\wedge\tau_U^+(b)}e^{-rt}p\diff t-\alpha e^{-r\tau_D^+(a)}\mathbbm{1}_{\left(\tau<\tau_D^+(a)<\tau_U^+(b)\right)}-ce^{-r\tau}\mathbbm{1}_{\left(\tau<\tau_D^+(a)\wedge\tau_U^+(b)\right)}\Bigg].\nonumber
\end{align}
The first summand on the right hand side does not depend on the stopping time and is equal to the basic contract with drawup contingency, that is, it equals $k(y,z,p)$ given in \eqref{kyzp}. On the other hand, the second summand depends on $\tau$ through the event $\{\tau<\tau_D^+(a)\wedge\tau_U^+(b)\}$. The result follows now by the strong Markov property:
\begin{align}
K(y,z,p)=k(y,z,p)+\sup\limits_{\tau\in\mathcal{T}}\mathbb{E}_{|y|z}&\Bigg[e^{-r\tau}\mathbbm{1}_{\left(\tau<\tau_D^+(a)\wedge\tau_U^+(b)\right)}\mathbb{E}_{|D_\tau |U_\tau}\Big[\int_{0}^{\tau_D^+(a)\wedge\tau_U^+(b)}e^{-rt}p\diff t\Big]\nonumber\\
&-\alpha e^{-r\tau_D^+(a)}\mathbbm{1}_{\left(\tau<\tau_D^+(a)<\tau_U^+(b)\right)}-ce^{-r\tau}\mathbbm{1}_{\left(\tau<\tau_D^+(a)\wedge\tau_U^+(b)\right)}\Bigg].
\end{align}
\end{proof}
First note that if $\tilde{k}(D_{\tau_D^-(\theta)},U_{\tau_D^-(\theta)})<0$ for all $\theta$, then
it is not optimal to terminate the contract and hence $\tau=\infty$.
To avoid this case, we assume from now on that there exists $\theta_0$
for which $\tilde{k}(D_{\tau_D^-(\theta_0)},U_{\tau_D^-(\theta_0)})>0$. We can rewrite this assumption as follows:
\begin{align}
\frac{p}{r}-c>\left(\frac{p}{r}+\alpha\right)\nu(\theta_0,y+z-\theta_0)+\frac{p}{r}\lambda(\theta_0,y+z-\theta_0)\geq 0\label{war2a}
\end{align}
for $y+z\geq a$, and
\begin{align}
\frac{p}{r}-c>\left(\frac{p}{r}+\alpha\right)\nu(\theta_0,y-x_0-\theta_0)+\frac{p}{r}\lambda(\theta_0,y-x_0-\theta_0)\geq 0\label{war2b}
\end{align}
for $y+z<a$, where $x_0$ satisfies $\tilde{k}(\theta_0,y-x_0-\theta_0)=\min\limits_{x\in (y-a,-z)}\tilde{k}(\theta_0,y-x-\theta_0)$.
Additionally, because of the presence of the indicator in \eqref{htau}, without loss of generality we can assume
that
\begin{equation*}
b-z>y-\theta.
\end{equation*}
To identify the value of the contract $K$ we will now compute the function $H$ defined in \eqref{fh}.
We will again use the ``guess and verify'' approach. A candidate for the optimal strategy is
\begin{equation}\label{tau*drugie}
\tau^*=\tau_D^-(\theta)\end{equation}
for some $\theta\in [0,a)$.
We denote
\begin{align}
h(y,z,p,\theta):=h_{\tau_D^-(\theta)}(y,z,p).\label{h1}
\end{align}
We will compute this function now.
Note that for $\theta>y$ we have
\begin{align}
h(y,z,p,\theta)=\tilde{k}(y,z,p).\label{h=k}
\end{align}
Moreover,
$$U_{\tau^-_D(\theta)}=X_{\tau^-_D(\theta)}-\underline{X}_{\tau^-_D(\theta)}\wedge (-z)=y-\theta-\underline{X}_{\tau^-_D(\theta)}\wedge (-z).$$
Thus, by considering two disjoint possible scenarios $\{\underline{X}_{\tau_D^-(\theta)}> -z\}$ and
$\{\underline{X}_{\tau_D^-(\theta)}\leq-z\}$,
the expectation in \eqref{htau} can be rewritten for $y\geq\theta$ as
\begin{align}
&h(y,z,p,\theta)=\tilde{k}(\theta,y-\theta+z,p)\mathbb{E}_{|y|z}\left[e^{-r\tau_D^-(\theta)};\tau_D^-(\theta)<\tau^+_D(a)\wedge\tau_U^+(b),\underline{X}_{\tau_D^-(\theta)}> -z\right]\nonumber\\
&\qquad+\mathbb{E}_{|y|z}\left[e^{-r\tau_D^-(\theta)}\tilde{k}(\theta,y-\theta-\underline{X}_{\tau_D^-(\theta)},p);\tau_D^-(\theta)<\tau^+_D(a)\wedge\tau_U^+(b),\underline{X}_{\tau_D^-(\theta)}\leq-z\right].\label{nowyrepp}
\end{align}
We will now analyze the event appearing in both the last expectations.
\begin{Prop}\label{propDtheta}
The following events are equivalent:
\begin{align}
\left\{\tau_D^-(\theta)<\tau_D^+(a)\wedge\tau_U^+(b),\ D_0=y,\ U_0=z\right\}&=
\left\{\underline{X}_{\tau_D^-(\theta)}>y-a,\ \underline{X}_{\tau_D^-(\theta)}\wedge (-z)>y-\theta -b\right\}.\label{rhs}
\end{align}
\end{Prop}
\begin{proof}
Note that the stopping time $\tau_D^-(\theta)$ is when the process $X$ hits $y-\theta$.
This means that $X$ cannot exceed level $y$ before $\tau_D^-(\theta)$
and therefore we have $\overline{X}_{\tau_D^-(\theta)}\vee y=y.$
Now, the first event $\{\underline{X}_{\tau_D^-(\theta)}>y-a\}$ on the right hand side of \eqref{rhs}
corresponds to the situation when $\tau_D^+(a)$ is after $\tau_D^-(\theta)$.
On the second event
$\{\underline{X}_{\tau_D^-(\theta)}\wedge (-z)>y-\theta -b\}$
on the right hand side of \eqref{rhs}, the drawup process $U$ attains level $b$ only after
the first passage of $y-\theta$ by $X$.
In this case too, $\tau_U^+(b)$ cannot be before $\tau_D^-(\theta)$.
This observation completes the proof.
\end{proof}

Proposition \ref{propDtheta} and \eqref{nowyrepp} give the following representation of the function $h$ defined formally in (\ref{h1}).
\begin{Lemma}\label{abovelemma}
For $y\geq\theta$ we have
\begin{align}
&h(y,z,p,\theta )=\tilde{k}(\theta,y+z-\theta)\frac{W^{(r)}\left((a-y)\wedge z\right)}{W^{(r)}\left(y-\theta+(a-y)\wedge z\right)}\mathbbm{1}_{(y-\theta<b-z)}
\nonumber\\
&\qquad+\mathbb{E}_{|y}\left[e^{-r\tau_D^-(\theta)}\tilde{k}(\theta,y-\theta-\underline{X}_{\tau_D^-(\theta)},p);(y-a)\vee(y-\theta-b)<\underline{X}_{\tau_D^-(\theta)}\leq -z\right]\mathbbm{1}_{((y-a)\vee(y-\theta-b)<-z)}.\label{htheta}
\end{align}
\end{Lemma}
Observe that in order to calculate (\ref{h1}) (or \eqref{htheta}) we only need to know the joint distribution of $\underline{X}_{\tau^-_D(\theta)}$ and $\tau_D^-(\theta)$. This can be derived using (\ref{distXD-}).

In order to satisfy \eqref{tau*drugie} we look for $\theta$ that maximizes the function $h(y,z,p,\theta)$. We denote
\begin{align}\label{tetadu*}
\theta^*=\inf\left\{ \theta\in [0,a):\ \frac{\partial}{\partial\theta}h(y,z,p,\theta)=0,\ \text{and }h(y,z,p,\varsigma)\leq h(y,z,p,\theta)\ \textrm{for all $\varsigma\geq 0$}\right\}.
\end{align}
Note that if there is no local maximum of $h$ on $[0,a)$, then $\tau_D^-(\theta)$ is not the optimal stopping time for the problem under study.
\begin{theorem}\label{th2}
Assume that \eqref{war2a} and \eqref{war2b} hold and there exists $\theta^*$ defined in \eqref{tetadu*}.
Then $\tau_D^-(\theta^*)$ with $\theta^*$ given by \eqref{tetadu*} is the optimal stopping rule solution for \eqref{htau} and the value of the contract \eqref{K} equals
$K(y,z,p)=k(y,z,p)+h(y,z,p,\theta^*)$ for $h(y,z,p,\theta^*)$ given in \eqref{h=k} and \eqref{htheta} and $k(y,z,p)$ given in \eqref{fdrawup}.
\end{theorem}
\begin{proof}
The optimization problem we are dealing with here is defined in \eqref{fh}.
We will again use Verification Lemma \ref{war} for $\Upsilon_t=(D_t, U_t)$, $\mathbb{B}=\mathbb{R}_+\times\mathbb{R}_+$,
$\tau_0=\tau_U^+(b)\wedge\tau_D^+(a)$, $V(\phi)=\tilde{k}(y,z,p)$ with $\phi=(y,z)$.
The proof is similar to the proof of Theorem \ref{th1}.
The proof of condition (i) of Verification Lemma \ref{war} follows in fact the same pattern with $\theta=y$ at the first step.

Now we will prove the second condition of Verification Lemma \ref{war}. Let $\delta:=\tau_D^+(a)\wedge\tau_U^+(b)\wedge\tau_D^-(\theta^*)$. Note that, using the strong Markov property from the definition of the function $h(y,z,p, \theta^*)$ given in \eqref{h1} and \eqref{htau} for $\Upsilon_t=(D_t, U_t)$, we obtain
\begin{align}
&\mathbb{E}\left[h(D_{\delta},U_{\delta},p,\theta^*)|\ \mathcal{F}_{t\wedge\delta}\right]=\mathbb{E}\left[\mathbb{E}_{|D_{\delta}|U_{\delta}}\left[e^{-r\tau_D^-(\theta^*)}\tilde{k}(D_{\tau_D^-(\theta^*)},U_{\tau_D^-(\theta^*)},p)\right]\big|\ \mathcal{F}_{t\wedge\delta}\right]\nonumber\\
&=\mathbb{E}\left[e^{-r(t\wedge\delta)}\mathbb{E}_{|D_{t\wedge\delta}|U_{t\wedge\delta}}\left[e^{-r\tau_D^-(\theta^*)}\tilde{k}(D_{\tau_D^-(\theta^*)},U_{\tau_D^-(\theta^*)},p)\right]\big|\ \mathcal{F}_{t\wedge\delta}\right]=e^{-r(t\wedge\delta)}h(D_{t\wedge\delta},U_{t\wedge\delta},p,\theta^*).\nonumber
\end{align}
Thus, we can conclude that the process
$$e^{-r(t\wedge\tau_D^+(a)\wedge\tau_U^+(b)\wedge\tau_D^-(\theta^*))}h(D_{t\wedge\tau_D^+(a)\wedge\tau_U^+(b)\wedge\tau_D^-(\theta^*)},U_{t\wedge\tau_D^+(a)\wedge\tau_U^+(b)\wedge\tau_D^-(\theta^*)},p,\theta^*)$$
is a martingale. Hence, for $y>\theta^*$,
\[\mathcal{A}_{(D,U)} h(y, z, p,\theta^*)-rh(y, z, p,\theta^*)=0,\]
where $\mathcal{A}_{(D,U)}$ is the full generator of the process $(D_t, U_t)$.
Moreover, the processes
$$e^{-r(t\wedge\tau_D^+(a))}\nu(D_{t\wedge\tau_D^+(a)},U_{t\wedge\tau_D^+(a)}),\qquad e^{-r(t\wedge\tau_U^+(b))}\lambda(D_{t\wedge\tau_U^+(b)},U_{t\wedge\tau_U^+(b)})$$
are also $\mathcal{F}_t$-martingales. Thus
$\mathcal{A}_{(D,U)}\nu(y,z)-r\nu(y,z)=0$ and $\mathcal{A}_{(D,U)}\lambda(y,z)-r\lambda(y,z)=0$.
Summing up, by (\ref{war2a}), \eqref{war2b} and (\ref{h=k}) for $y\in(0,\theta^*)$ and $z\in(0,b)$ we have
\begin{eqnarray}
\mathcal{A}_{(D,U)}h(y,z,p,\theta^*)-rh(y,z,p,\theta^*)=
\mathcal{A}_{(D,U)}\tilde{k}(y,z,p)-r\tilde{k}(y,z,p)=
-r\left(\frac{p}{r}-c\right)\leq 0.\nonumber
\end{eqnarray}
This completes the proof.
\end{proof}

\subsection{A strikingly simple case when $a=b$}\label{a=b}
The case where $a=b$ corresponds to the situation when the contract pays the compensation when
the drawdown process exceeds level $a$ or expires when the drawup reaches $a$.
In this subsection we will find the function $\lambda(y,z)$ and $\nu(y,z)$ appearing in
\eqref{fdrawup} and \eqref{p*drawup} (hence also in \eqref{kdecompose} and \eqref{htheta}).

The simplicity of this case follows from the fact that we can divide the problem into two easy sub-cases.
The first one is when $a\leq z+y$. Then $\lambda$ and $\nu$ can be identified by using the two-sided exit formulas \eqref{twosided1}-\eqref{twosided2}. Indeed,
\begin{align}
\lambda(y,z)&=\mathbb{E}_{|y|z}\left[e^{-r\tau_U^+(a)};\tau_U^+(a)<\tau_D^+(a)\right]=\mathbb{E}\left[e^{-r\tau^+_{a-z}};\tau^+_{a-z}<\tau^-_{y-a}\right]\nonumber\\
&=\frac{W^{(r)}(a-y)}{W^{(r)}(2a-y-z)}\label{lambda1}
\end{align}
and
\begin{align}
\nu(y,z)&=\mathbb{E}_{|y|z}\left[e^{-r\tau_D^+(a)};\tau_D^+(a)<\tau_U^+(a)\right]=\mathbb{E}\left[e^{-r\tau^-_{y-a}};\tau^-_{y-a}<\tau^+_{a-z}\right]\nonumber\\
&=Z^{(r)}(a-y)-Z^{(r)}(2a-y-z)\frac{W^{(r)}(a-y)}{W^{(r)}(2a-y-z)}.\label{nu1}
\end{align}
The second case is when $a>z+y$.
For this case the following identity is crucial:
\begin{align}
\left\{\tau_U^+(a)\in dt,\ \tau_D^+(a)>t,X_t\in dx\right\}=\left\{\tau^+_x\in dt,\underline{X}_t\wedge(-z)\in d(x-a)\right\},\label{usingthis}
\end{align}
which holds for any $x\in(y,a-z]$ (see \cite[Eq. (46)]{olimpia2010} for details).
Using \eqref{usingthis} we observe that
\begin{align}
\lambda(y,z)&=\mathbb{E}_{|y|z}\left[e^{-r\tau^+_U(a)};\tau^+_U(a)<\tau_D^+(a)\right]\nonumber\\
&=\int_0^\infty\int_0^\infty e^{-rt}\mathbb{P}_{|y|z}\left(\tau^+_U(a)\in \diff t;\tau_D^+(a)>t;X_t\in \diff x\right)\nonumber\\
&=\int_0^\infty\int_0^\infty e^{-rt}\mathbb{P}\left(\tau^+_x\in \diff t;\underline{X}_{\tau_x^+}\wedge(-z)\in \diff (x-a),x\in(y,a-z)\right)\nonumber\\
&=\int_y^{a-z}\mathbb{E}\left[e^{-r\tau^+_x};\underline{X}_{\tau_x^+}\wedge(-z)\in \diff x-a\right]\nonumber\\
&=\int_y^{a-z}\mathbb{E}\left[e^{-r\tau^+_x};\underline{X}_{\tau_x^+}\in \diff (x-a)\right]+\mathbb{E}\left[e^{-r\tau^+_{a-z}};\underline{X}_{\tau^+_{a-z}}>-z\right]\nonumber\\
&=\int_y^{a-z}\frac{\partial}{\partial a}\frac{W^{(r)}(a-x)}{W^{(r)}(a)}\diff x+\mathbb{E}\left[e^{-r\tau^+_{a-z}};\underline{X}_{\tau^+_{a-z}}>-z\right]\nonumber\\
&=\frac{W^{(r)}(a-y)}{W^{(r)}(a)}-\frac{1}{r}\frac{W^{\prime (r)}(a)}{(W^{(r)}(a))^2}\left(Z^{(r)}(a-y)-Z^{(r)}(z)\right),\label{lambda}
\end{align}
where we have used the fact that
\begin{align}
\mathbb{E}\left[e^{-r\tau^+_x};\underline{X}_{\tau^+_x}\in -du\right]=\frac{\partial}{\partial u}\frac{W^{(r)}(u)}{W^{(r)}(x+u)}(-du).
\end{align}
The function $\nu$ can be calculated by using the formula for the Laplace transform $\xi$ of $\tau_D^+(a)$ given in (\ref{ksi})
and the above expression for $\lambda$. Indeed,
\begin{align}
\nu(y,z)&=\mathbb{E}_{|y|z}\left[e^{-r\tau^+_D(a)};\tau^+_D(a)<\tau^+_U(a)\right]\nonumber\\
&=\mathbb{E}_{|y|z}\left[e^{-r\tau^+_D(a)}\right]-\mathbb{E}_{|y|z}\left[e^{-r\tau^+_D(a)};\tau^+_U(a)<\tau^+_D(a)\right]\nonumber\\
&=\mathbb{E}_{|y|}\left[e^{-r\tau^+_D(a)}\right]-\mathbb{E}_{|y|z}\left[e^{-r\tau^+_U(a)};\tau^+_U(a)<\tau^+_D(a)\right]\mathbb{E}\left[e^{-r\tau_D^+(a)}\right]\nonumber\\
&=Z^{(r)}(a-y)-rW^{(r)}(a-y)\frac{W^{(r)}(a)}{W^{\prime (r)}(a)}-\lambda(y,z)\left(Z^{(r)}(a)-rW^{(r)}(a)\frac{W^{(r)}(a)}{W^{\prime (r)}(a)}\right)\nonumber\\
&=Z^{(r)}(z)-Z^{(r)}(a)\frac{W^{(r)}(a-y)}{W^{(r)}(a)}+\frac{1}{r}Z^{(r)}(a)\frac{W^{\prime (r)}(a)}{(W^{(r)}(a))^2}\left(Z^{(r)}(a-y)-Z^{(r)}(z)\right).\label{nu}
\end{align}
Identification of $\nu$ and $\lambda$ allows us also to calculate the function $h$ appearing in the value function \eqref{K}
given in Theorem \ref{th2}.
Precisely, by \eqref{h=k} for $\theta\leq y$ we have
\begin{align}
h(y,z,p,\theta)=\tilde{k}(y,z,p),\nonumber
\end{align}
where $\tilde{k}$ can be identified using (\ref{fdrawup}) and \eqref{ktilde}.
For $\theta>y$ we have
\begin{align}
h(y,z,p,\theta)=&\mathbb{E}_{|y}\left[e^{-r\tau_D^-(\theta)}\tilde{k}(\theta,y-\theta-\underline{X}_{\tau_D^-(\theta)});\ y-a<\underline{X}_{\tau_D^-(\theta)}<-z\right]\nonumber\\
&+\tilde{k}(\theta,y+z-\theta)\mathbb{E}_{|y}\left[e^{-r\tau_D^-(\theta)};\ \underline{X}_{\tau_D^-(\theta)}>(y-a)\vee (-z)\right]\mathbbm{1}_{(a>y+z-\theta)}\nonumber\\
=&\int_z^{a-y}\tilde{k}(\theta,y-\theta+\phi)\frac{\partial}{\partial\phi}\frac{W^{(r)}(\phi)}{W^{(r)}(y-\theta+\phi)}\diff\phi\mathbbm{1}_{(a>y+z)}\nonumber\\
&+\tilde{k}(\theta,y+z-\theta)\frac{W^{(r)}((a-y)\wedge z)}{W^{(r)}((a-y)\wedge z +y-\theta)}\mathbbm{1}_{(a>y+z-\theta)}.\nonumber
\end{align}
Now, by integration by parts we have
\begin{align}
&h(y,z,p,\theta)=\tilde{k}(\theta,a-\theta)\frac{W^{(r)}(a-y)}{W^{(r)}(a-\theta)}\mathbbm{1}_{(a>y+z)}-\tilde{k}(\theta,y+z-\theta)\frac{W^{(r)}(z)}{W^{(r)}(y+z-\theta)}\mathbbm{1}_{(a>y+z)}\nonumber\\
&\qquad+\tilde{k}(\theta,y+z-\theta)\frac{W^{(r)}((a-y)\wedge z)}{W^{(r)}((a-y)\wedge z+y-\theta)}\mathbbm{1}_{(a>y+z-\theta)}\nonumber\\
&\qquad-\int_z^{a-y}\left[-\left(\frac{p}{r}+\alpha\right)\left(1-\frac{1}{r}Z^{(r)}\frac{W^{\prime (r)}(a)}{(W^{(r)}(a))^2}\right)-\frac{p}{r}\frac{1}{r}\frac{W^{\prime (r)}(a)}{(W^{(r)}(a))^2}\right]rW^{(r)}(\phi)\diff\phi\mathbbm{1}_{(a>y+z)}\nonumber\\
=&\left\{ \begin{array}{l}
\left[-\left(\frac{p}{r}+\alpha\right)\left(1-\frac{1}{r}Z^{(r)}\frac{W^{\prime (r)}(a)}{(W^{(r)}(a))^2}\right)-\frac{p}{r}\frac{1}{r}\frac{W^{\prime (r)}(a)}{(W^{(r)}(a))^2}\right]\left(Z^{(r)}(a-y)-Z^{(r)}(z)\right)\\
\begin{array}{ll}
\qquad +\tilde{k}(\theta,a-\theta)\frac{W^{(r)}(a-y)}{W^{(r)}(a-\theta)} &\qquad\qquad\qquad\qquad\qquad\qquad\textrm{for }a>y+z,\\
\tilde{k}(\theta,y+z-\theta)\frac{W^{(r)}(a-y)}{W^{(r)}(a-\theta)}&\qquad\qquad\qquad\qquad\qquad\qquad\textrm{for }y+z-\theta<a<y+z,\\
0&\qquad\qquad\qquad\qquad\qquad\qquad\textrm{for }a<y+z-\theta.\label{ha=b}
\end{array}
\end{array}\right.
\end{align}

\section{Numerical analysis}\label{sec:examples}
In this section we analyze numerically all the insurance contracts under study and
check their dependence on selected parameters of the model.
We focus on two classical spectrally negative risk L\'evy processes. The first one is the linear Brownian motion:
\begin{equation}
X_t=\mu t+\sigma B_t,\label{linearBrownian}
\end{equation}
where $B_t$ is a standard Brownian motion. The second process we analyze is the classical Cram\'er-Lundberg process with exponential jumps:
\begin{equation}\label{CL}
X_t=\mu t-\sum_{i=1}^{N_t}{\eta_i},\end{equation}
where $\eta_i$ are i.i.d. exponentially distributed random variables with parameter $\rho >0$ and $N_t$
is an independent Poisson process with intensity $\beta>0$.

We express all the main quantities and all contract values in terms of the scale functions defined in \eqref{Wq} and \eqref{Zq}.
From \cite{kyprianoua} it follows that  the scale functions for the Brownian motion with drift \eqref{linearBrownian} take the following form:
\begin{align}
W^{(r)}(u) =&\frac{2}{\sigma^2\Xi}e^{-\frac{\mu}{\sigma^2}u}\sinh \left(\Xi u\right),\label{bmW}\\
Z^{(r)}(u)=&e^{-\frac{\mu}{\sigma^2}u}\left(\cosh\left(\Xi u\right)+\frac{\mu}{\Xi\sigma^2}\sinh\left(\Xi u\right)\right),\label{bmZ}
\end{align}
where
\begin{align}
\Xi=\frac{\sqrt{\mu^2+2r\sigma^2}}{\sigma^2}.\nonumber
\end{align}
Similarly, for the Cram\'er-Lundberg process \eqref{CL},
\begin{align}
W^{(r)}(u) =&\frac{e^{\Phi(r) u}}{\psi^\prime(\Phi(r))}+\frac{e^{\zeta u}}{\psi^\prime(\zeta)}, \label{clW}\\
Z^{(r)}(u)=&1+r\frac{e^{\Phi(r) u}-1}{\Phi(r) \psi^\prime (\Phi(r) )}+r\frac{e^{\zeta u}-1}{\zeta\psi^\prime(\zeta)}, \label{clZ}
\end{align}
where
\begin{align}
&\Phi(r) =\frac{1}{2\mu}\left ((\beta+r-\mu\rho)+\sqrt{(\beta+r-\mu\rho)^2+4r\mu\rho}\right ),\nonumber\\
&\zeta =\frac{1}{2\mu}\left ((\beta+r-\mu\rho)-\sqrt{(\beta+r-\mu\rho)^2+4q\mu\rho}\right ),\nonumber\\
&\psi^\prime(\phi)=\mu - \frac{\beta\rho}{(\rho+\phi)^2}\nonumber
\end{align}
and $\psi(\cdot)$ is the Laplace exponent given in \eqref{psi}.

In this section we analyze the influence of selected parameters of our model
on the prices and the stopping rules of the insurance contracts under consideration.
To simplify the above comparison, we order the numerical analysis according to the order of appearance of these contracts in this paper.

\subsection{Fair premium for drawdown insurance}
We start from the contract \eqref{f} using \eqref{dodane1}.
Let $X_t$ be a linear Brownian motion \eqref{linearBrownian}.
From Proposition \ref{Ksi} we have
\begin{eqnarray*}
\xi (y) = e^{-\frac{\mu}{\sigma^2}(a-y)}\frac{\Xi\cosh(\Xi y)-\frac{\mu}{\sigma^2}\sinh(\Xi y)}{\Xi\cosh(\Xi a)-\frac{\mu}{\sigma^2}\sinh(\Xi a)}.
\end{eqnarray*}
This leads to a formula for the value function $f(y,p)$ given in \eqref{f} and the expression for the fair premium $p^*$ given in \eqref{p*}:
\begin{align}
f(y,p)=\left(\frac{p}{r}+\alpha\right)e^{-\frac{\mu}{\sigma^2}(a-y)}\frac{\Xi\cosh(\Xi y)-\frac{\mu}{\sigma^2}\sinh(\Xi y)}{\Xi\cosh(\Xi a)-\frac{\mu}{\sigma^2}\sinh(\Xi a)}-\frac{p}{r},\nonumber\\
p^*=\frac{r\alpha(\Xi\cosh(\Xi y)-\frac{\mu}{\sigma^2}\sinh(\Xi y))}{\Xi(\cosh(\Xi a)-\cosh(\Xi y))-\frac{\mu}{\sigma^2}(\sinh(\Xi a)-\sinh(\Xi y))}.\nonumber
\end{align}
In Figure \ref{p*bm} we depict the fair premium $p^*$ depending on the starting drawdown position $D_0=y$.
\begin{figure}[!ht]
\center{
\includegraphics[width=0.4\textwidth]{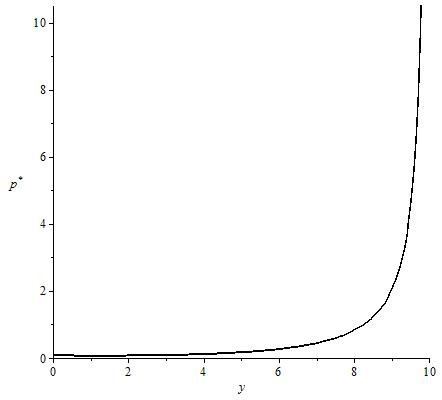}
\caption{\footnotesize{The value $p^*$ for the drawdown insurance contract for Brownian motion with drift.
Parameters: $r=0.01, \mu=0.03, \sigma=0.4, \alpha=100, a=10$.}}\label{p*bm}}
\end{figure}
\\
\indent Similar calculations are done for the Cram\'er-Lundberg process given in \eqref{CL}. In particular, we have
\begin{align}
\xi (y) =& 1+r\frac{e^{\Phi(r) (a-y)}-1}{\Phi(r)\psi^\prime(\Phi(r) )}+r\frac{e^{\zeta (a-y)}-1}{\zeta\psi^\prime(\zeta)}-r\left ( \frac{e^{\Phi(r) (a-y)}}{\psi^\prime(\Phi(r) )}+\frac{e^{\zeta (a-y)}}{\psi^\prime(\zeta)}\right ) \nonumber \\
&\cdot\frac{\psi^\prime(\zeta)e^{\Phi(r) a}+\psi^\prime(\Phi(r) )e^{\zeta a}}{\Phi(r) \psi^\prime(\zeta)e^{\Phi(r) a}+\zeta\psi^\prime(\Phi(r) )e^{\zeta a}}=: c_0 + c_{\Phi(r)}e^{\Phi(r) (r)(a-y)}+c_{-\zeta}e^{-\zeta (a-y)},\nonumber
\end{align}
where
\begin{align}
c_0=&1-\frac{r}{\Phi(r) \psi^\prime(\Phi(r) )}-\frac{r}{\zeta\psi^\prime(\zeta)},\nonumber\\
c_{\Phi(r)}=&\frac{r}{\Phi(r) \psi^\prime(\Phi(r) )}-\frac{r}{\psi^\prime(\Phi(r) )}\frac{\psi^\prime(\zeta)e^{\Phi(r) a}+\psi^\prime(\Phi(r) )e^{\zeta a}}{\Phi(r) \psi^\prime(\zeta)e^{\Phi(r) a}+\zeta\psi^\prime(\Phi(r) )e^{\zeta a}},\nonumber\\
c_{\zeta}=&\frac{r}{\zeta\psi^\prime(\zeta)}-\frac{r}{\psi^\prime(\zeta)}\frac{\psi^\prime(\zeta)e^{\Phi(r) a}+\psi^\prime(\Phi(r) )e^{\zeta a}}{\Phi(r) (r)\psi^\prime(\zeta)e^{\Phi(r) a}+\zeta\psi^\prime(\Phi(r) )e^{\zeta a}}.\nonumber
\end{align}
The contract value $f(y,p)$ given in \eqref{f} and the fair premium $p^*$ given in \eqref{p*} are
\begin{align}
f(y,p)&=\left(\frac{p}{r}+\alpha\right)\left(c_0 + c_{\Phi(r)}e^{\Phi(r) (a-y)}+c_{\zeta}e^{\zeta (a-y)}\right)-\frac{p}{r},\nonumber\\
p^*&=\frac{r\alpha(c_0 + c_{\Phi(r)}e^{\Phi(r) (a-y)}+c_{\zeta}e^{\zeta (a-y)})}{1-c_0 - c_{\Phi(r)}e^{\Phi(r) (a-y)}-c_{\zeta}e^{\zeta (a-y)}}.\nonumber
\end{align}
Figure \ref{p*cl} describes the dependence of the fair premium $p^*$ on the starting drawdown $D_0=y$ in this case.
\begin{figure}[!ht]
\center{
\includegraphics[width=0.4\textwidth]{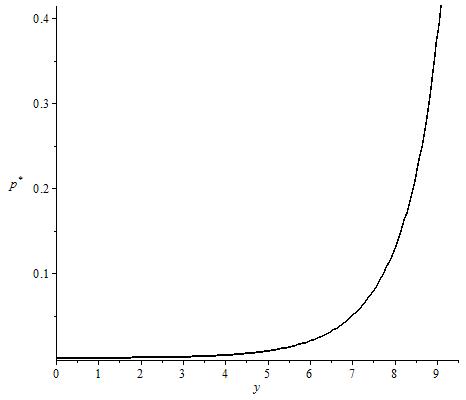}
\caption{\footnotesize{The value $p^*$ for the drawdown insurance contract for the Cram\'er--Lundberg model.
Parameters: $r=0.01, \mu=0.05, \beta=0.1, \rho=2.5, \alpha=100, a=10$.}}\label{p*cl}}
\end{figure}

Note that for both, the linear Brownian motion \eqref{linearBrownian} and
the Cram\'er-Lundberg model \eqref{CL}, the shapes of the fair premium are
quite similar. Moreover, increasing the starting drawdown
$D_0=y$ may rapidly increase the value $p^*$.
In fact, for the Brownian motion with drift, the value of $p^*$ tends to $\infty$ as $y\uparrow a$.
This follows from the fact that for the linear Brownian motion we have
 $W^{(r)}(0)=0$ and the denominator in the expression for the fair premium $p^*$ goes to $0$ as
$y\uparrow a$.
For the Cram\'er--Lundberg process \eqref{CL} this is not the case, though.
Indeed, for this process of bounded variation we have $W^{(r)}(0)>0$ and $\lim\limits_{y\rightarrow a^-}\xi(y)\neq 1$, so
the denominator in the formula for $p^*$ does not converge to $0$ as $y\uparrow a$.

\subsection{Cancellable drawdown insurance}
We now price the contract $F(y,p)=f(y,p) +g(y,p,\theta^*)$ defined in \eqref{F} and identified in Theorem \ref{th1}.
We also calculate the optimal stopping rule
$\tau^*$ given in \eqref{tau*pierwsze} for $\theta^*$ defined in \eqref{teta*}.
Thanks to the numerical results presented in the previous subsection,
it suffices to find the function $g(y,p,\theta^*)$ and $\theta^*$.

For the linear Brownian motion model \eqref{linearBrownian} we can write an explicitly formula for $g$.
In particular, for $\theta<y$,
\begin{align}
&g(y,p,\theta)=\tilde{f}(\theta,p)\frac{W^{(r)}(a-y)}{W^{(r)}(a-\theta)}=\left(\frac{p}{r}-c\right)\frac{e^{-\frac{\mu}{\sigma^2}(a-y)}\sinh(\Xi(a-y))}{e^{-\frac{\mu}{\sigma^2}(a-\theta)}\sinh(\Xi(a-\theta))}\nonumber\\
&\qquad-\left(\frac{p}{r}+\alpha\right)\frac{e^{-\frac{\mu}{\sigma^2}(a-y)}\sinh(\Xi (a-y))\left(\Xi\cosh(\Xi\theta)-\frac{\mu}{\sigma^2}\sinh(\Xi\theta)\right)}{\sinh(\Xi(a-\theta))\left(\Xi\cosh(\Xi k)-\frac{\mu}{\sigma^2}\sinh(\Xi k)\right)},\nonumber
\end{align}
and for $\theta\geq y$ we have
\begin{align}
g(y,p,\theta)=\tilde{f}(y,p)=-\left(\frac{p}{r}+\alpha\right)e^{-\frac{\mu}{\sigma^2}(a-y)}\frac{\Xi\cosh(\Xi y)-\frac{\mu}{\sigma^2}\sinh(\Xi y)}{\Xi\cosh(\Xi a)-\frac{\mu}{\sigma^2}\sinh(\Xi a)}+\frac{p}{r}-c.
\end{align}
\begin{figure}[!ht]
\center{
\subfloat{\includegraphics[width=0.45\textwidth]{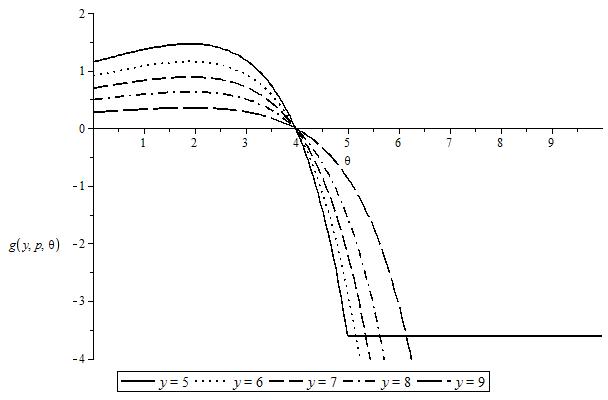}}
\quad
\subfloat{\includegraphics[width=0.45\textwidth]{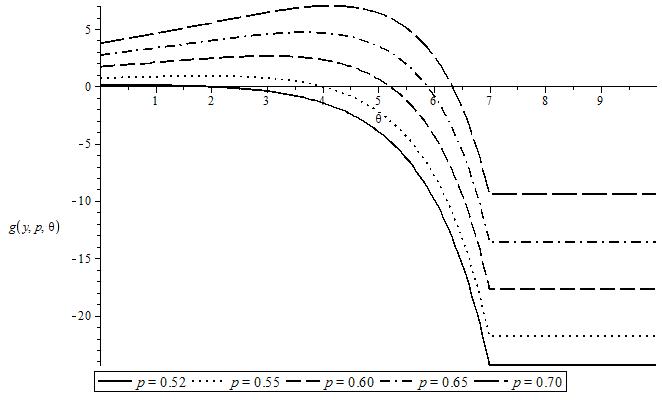}}
\caption{\footnotesize{The value $g(y,p,\theta)$ for Brownian motion with drift in dependence on $\theta$
for different levels of $y$ (left) and $p$ (right).
Parameters: $r=0.01, \mu=0.03, \sigma=0.4, \alpha=100, c=50, a=10, p=0.55, y=7$.}}\label{gbm}}
\end{figure}

Figure \ref{gbm} depicts the dependence of $g(y,p,\theta)$ on $\theta$.
The straight line piece at the end of the
graph follows from the fact that by \eqref{g=f} the function $g$ is constantly equal to $\tilde{f}(y,p)$ for $\theta\geq y$.
We are looking for $\theta^*$ that maximizes $g$. If $\theta^*$ is at the beginning
of the straight line piece, the investor should not take this insurance contract.
This is an extreme case. In fact, most natural are the contracts when $\theta^*$ is between zero and the beginning of the straight line piece.
Note also that by \eqref{g} the and definition of the optimal $\theta^*$ given in \eqref{teta*}
the level $\theta^*$ does not depend on the starting position of the drawdown $D_0=y$ as long as
$y>\theta^*$. In particular, for our set of parameters, for a constant premium, we have $\theta^*\approx 2$.

Figure \ref{gbm} also shows the dependence of $g(y,p,\theta )$ on the premium rate $p$.
Note that higher premium rate produces higher values of $g$, hence also higher values of the
insurance contract $F$ given in \eqref{F}.

For the Cram\'er--Lundberg model \eqref{CL} the function $g$ takes the form
\begin{align}
g(y,p,\theta)=&\tilde{f}(\theta,p)\frac{W^{(r)}(a-y)}{W^{(r)}(a-\theta)}=\frac{\psi^\prime(\zeta)e^{\Phi(r)(a-y)}+\psi^\prime(\Phi(r))e^{\zeta(a-y)}}{\psi^\prime(\zeta)e^{\Phi(r)(a-\theta)}+\psi^\prime(\Phi(r))e^{\zeta(a-\theta)}}\nonumber\\
&\cdot\left[-\left(\frac{p}{r}+\alpha\right)\left(c_0+c_\Phi(r) e^{\Phi(r)(a-\theta)}+c_{\zeta}e^{\zeta(a-\theta)}\right)+\left(\frac{p}{r}-c\right)\right]\nonumber
\end{align}
for $\theta<y$,
and
\begin{align}
g(y,p,\theta)=&\tilde{f}(y,p)=-\left(\frac{p}{r}+\alpha\right)\left(c_0 + c_{\Phi(r)}e^{\Phi(r) (a-y)}+c_{\zeta}e^{\zeta (a-y)}\right)+\frac{p}{r}-c\nonumber
\end{align}
for $\theta\leq y$.
\begin{figure}[!ht]
\center{
\subfloat{\includegraphics[width=0.45\textwidth]{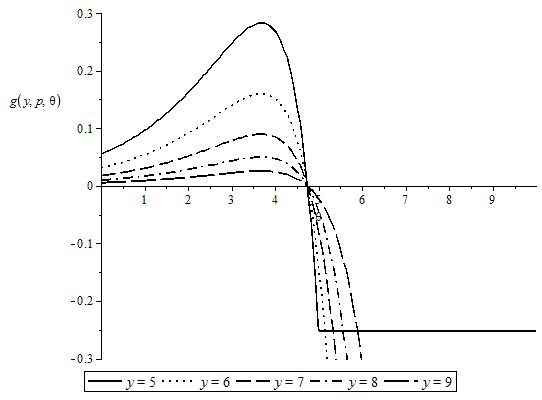}}
\quad
\subfloat{\includegraphics[width=0.45\textwidth]{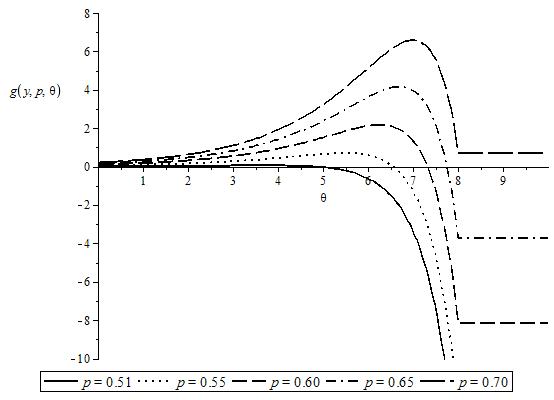}}
\caption{\footnotesize{The value $g(y,p,\theta)$ for the Cram\'er--Lundberg model in dependence on $\theta$ for different levels of $y$ (left) and $p$ (right).
Parameters: $r=0.01, \mu=0.05, \beta=0.1, \rho=2.5, a=10, \alpha=100, c=50, p=0.51, y=8$.}}\label{gcl}}
\end{figure}

Figure \ref{gcl} describes the behavior of $g(y,p,\theta )$ in dependence on the stopping level $\theta$,
hence identifying also the optimal one.

\subsection{Fair premium for a drawup contingency when $a>b$}
We will now investigate numerically the insurance contract
\eqref{kyzp} which provides protection from any specified drawdown of size $a$
with a certain drawup contingency of size $b$.
By Theorem \ref{thm2} it suffices to calculate the functions $\lambda$ and $\nu$ in order to identify the contract.
As mentioned in Section \ref{sec:drawup}, to do so, we will calculate these functions using
Corollary \ref{lambdanu} and numerically inverting the Laplace transform \eqref{m2}.

It is worth mentioning that in the case of linear Brownian motion \eqref{linearBrownian} there exists an alternative
(to inverting the Laplace transform) method of identifying $\lambda$ and $\nu$.
From Corollary \ref{lambdanu} the formulas for $\nu$ and $\lambda$ reduce to \eqref{twosided1}-\eqref{twosided2}. That is, for $a\leq y+z$,
\begin{align}
\nu(y,z)=Z^{(r)}(a-y)-Z^{(r)}(a+b-y-z)\frac{W^{(r)}(a-y)}{W^{(r)}(a+b-y-z)}\nonumber
\end{align}
and
\begin{align}
\lambda(y,z)=\frac{W^{(r)}(a-y)}{W^{(r)}(a+b-y-z)}.\nonumber
\end{align}
In order to identify $\lambda$ and $\nu$
when $a>y+z$, one can observe that $\widehat{X}_t:=-X_t$
is a linear Brownian motion with drift $-\mu$.
Then also $\widehat{U}_t=D_t$ and $\widehat{D}_t=U_t$.
Replacing $X$ by $\widehat{X}$ we can use the Laplace transform of $\widehat{D}$
given in \cite{Mijatovic1} and calculate the exact formulas for $\nu$ and $\lambda$.
Precisely, let $\widehat{W}^{(r)}$ and $\widehat{Z}^{(r)}$ be the scale functions for $\widehat{X}$ defined in \eqref{Wq} and \eqref{Zq}.
Then, for $a>y+z$,
\begin{align}
&\nu(y,z)=\frac{\widehat{W}^{(r)}(b)}{\widehat{W}^{\prime (r)}(b)}\frac{\sigma^2}{2}\left[\frac{(\widehat{W}^{\prime (r)}(b))^2}{\widehat{W}^{(r)}(b)}-\widehat{W}^{\prime (r)}(b)\right]e^{-(a-b\vee (y+z))\frac{\widehat{W}^{\prime (r)}(b)}{\widehat{W}^{(r)}(b)}}Z^{(r)}(b)\nonumber\\
&\qquad -\frac{1}{r}\frac{1}{\widehat{W}^{(r)}(b)}\frac{\sigma^2}{2}\left[\frac{(\widehat{W}^{\prime (r)}(b))^2}{\widehat{W}^{(r)}(b)}-\widehat{W}^{\prime (r)}(b)\right]\left(\widehat{Z}^{(r)}(b-z)-\widehat{Z}^{(r)}(y)\right)e^{-(a-b)\frac{\widehat{W}^{\prime (r)}(b)}{\widehat{W}^{(r)}(b)}}Z^{(r)}(b)\mathbbm{1}_{(b>y+z)}\nonumber
\end{align}
and
\begin{align}
&\lambda(y,z)=\frac{1}{r}\frac{1}{\widehat{W}^{(r)}(b)}\frac{\sigma^2}{2}\left[\frac{(\widehat{W}^{\prime (r)}(b))^2}{\widehat{W}^{(r)}(b)}-\widehat{W}^{\prime (r)}(b)\right]\left(\widehat{Z}^{(r)}(b\wedge(a-z))-\widehat{Z}^{(r)}(y)\right)e^{-(a-b)\frac{\widehat{W}^{\prime (r)}(b)}{\widehat{W}^{(r)}(b)}}\mathbbm{1}_{(b>y)}\nonumber\\
&\qquad+\frac{\widehat{W}^{(r)}(b)}{\widehat{W}^{\prime (r)}(b)}\frac{\sigma^2}{2}\left[\frac{(\widehat{W}^{\prime (r)}(b))^2}{\widehat{W}^{(r)}(b)}-\widehat{W}^{\prime (r)}(b)\right]\left(e^{-z\frac{\widehat{W}^{\prime (r)}(b)}{\widehat{W}^{(r)}(b)}}-e^{-(a-b\vee y)\frac{\widehat{W}^{\prime (r)}(b)}{\widehat{W}^{(r)}(b)}}\right)\mathbbm{1}_{(a>z+b)}+\frac{W^{(r)}(z)}{W^{(r)}(b)}.\nonumber
\end{align}

Using the above expressions we can find the fair premium $p^*$ defined in \eqref{p*drawup}:
\begin{align}
p^*=\frac{r\alpha\nu (y,z)}{1-\lambda(y,z)-\nu(y,z)},\nonumber
\end{align}
and analyze the influence of the starting position of the drawdown and drawup processes
on $p^*$. Figure \ref{P*bm} shows this relation.

\begin{figure}[!ht]
\center{
\subfloat{\includegraphics[width=0.4\textwidth]{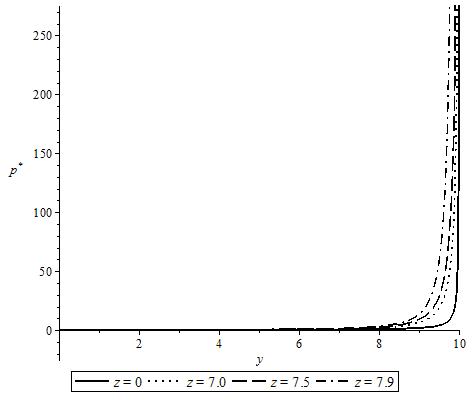}}
\quad
\subfloat{\includegraphics[width=0.4\textwidth]{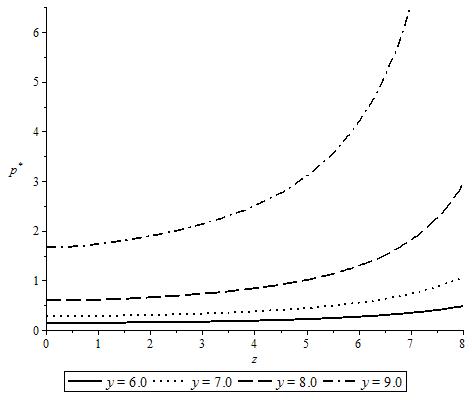}}
\caption{\footnotesize{The value $p^*$ for the drawup contingency contract for Brownian motion with drift for different starting positions of drawup $z$ (left) and drawdown $y$ (right). Parameters: $r=0.01, \mu=0.03, \sigma=0.4, \alpha=100, a=10, b=8$.}}\label{P*bm}}
\end{figure}

From Figure \ref{P*bm} we can deduce the same observation as for the basic drawdown contract.
For the linear Brownian motion \eqref{linearBrownian}
the value of $p^*$ tends to $\infty$ as $y\uparrow a$. This is because $\lim\limits_{y\rightarrow a^-}\nu(y,z) = 1$ and $\lim\limits_{y\rightarrow a^-}\lambda(y,z) = 0$,
and so, the denominator in the formula \eqref{p*drawup} for $p^*$  converges to $0$.

\subsection{Fair premium for drawup contingency when $a=b$}\label{sec:previous}
Here we analyze the special case $a=b$ presented in Subsection \ref{a=b}.
This time we use identities \eqref{lambda1}, \eqref{nu1}, \eqref{lambda} and \eqref{nu}
to compute the contract value.

Therefore, by the expression \eqref{p*drawup} for $p^*$,
\begin{align}
p^*=&\frac{r\alpha\nu(y,z)}{1-\lambda(y,z)-\nu(y,z)}=\frac{r\alpha\left(Z^{(r)}(a-y)-Z^{(r)}(2a-y-z)\frac{W^{(r)}(a-y)}{W^{(r)}(2a-y-z)}\right)}{1-Z^{(r)}(a-y)+\frac{W^{(r)}(a-y)}{W^{(r)}(2a-y-z)}\left(Z^{(r)}(2a-y-z)-1\right)}\label{P1}
\end{align}
for $a\leq y+z$,
and
\begin{align}
p^*=&\frac{r\alpha\left(Z^{(r)}(z)-Z^{(r)}(a)\frac{W^{(r)}(a-y)}{W^{(r)}(a)}+\frac{1}{r}Z^{(r)}(a)\frac{W'^{(r)}(a)}{(W^{(r)}(a))^2}\left(Z^{(r)}(a-y)-Z^{(r)}(z)\right)\right)}{1-Z^{(r)}(z)-\left(Z^{(r)}(a)-1\right)\left(\frac{1}{r}\frac{W'^{(r)}(a)}{(W^{(r)}(a))^2}(Z^{(r)}(a-y)-Z^{(r)}(z))-\frac{W^{(r)}(a-y)}{W^{(r)}(a)}\right)}\label{P2}
\end{align}
for $a>y+z$.
Using the formulas (\ref{clW})-(\ref{clZ}) for the scale functions
for the Cram\'er--Lundberg model \eqref{CL}, one can find
the dependence of $p^*$ on the initial starting/historical
positions of the drawdown $D_0=y$ and drawup $U_0=z$.
This dependence is depicted in Figure \ref{P*cl}.
\begin{figure}[!ht]
\center{
\subfloat{\includegraphics[width=0.4\textwidth]{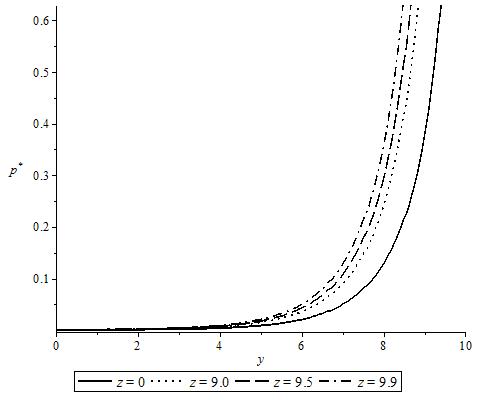}}
\quad
\subfloat{\includegraphics[width=0.4\textwidth]{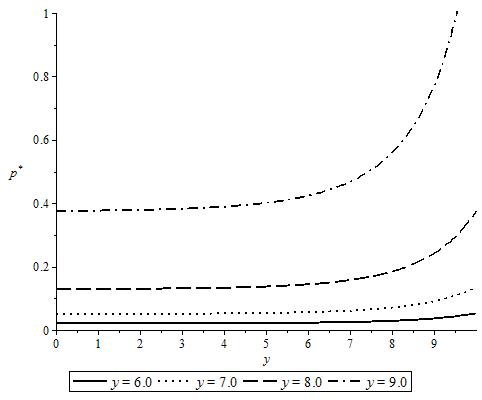}}
\caption{\footnotesize{The value $p^*$ for the drawup contingency contract
for the Cram\'er--Lundberg model for different starting positions of drawup $z$ (left)
and drawdown $y$ (right).
Parameters: $r=0.01, \mu=0.05, \beta=0.01, \rho=2.5, a=10, \alpha=100$.}}\label{P*cl}}
\end{figure}
Note that, similarly to the drawdown contract without drawup constraints,
the fair premium $p^*$ does not tend to $\infty$ as $y\uparrow a$
for the Cram\'er--Lundberg process \eqref{CL}.
This is a consequence of the fact that for the Cram\'er--Lundberg
process we have $W^{(r)}(0)>0$ and hence $\lim\limits_{y\rightarrow a^-}(\nu(y,z)+\lambda(y,z))\neq 1$,
and the denominator in \eqref{p*drawup} does not converge to $0$.

\subsection{Cancellable drawup contingency for $a>b$}
We continue the numerical analysis by adding cancellability and
by considering the insurance contract \eqref{K}.
By Theorem \ref{th2} it suffices to find $h(y,z,p,\theta^*)$
for $h$ given in \eqref{h=k} and \eqref{htheta} and
for the optimal level $\theta^*$ defined in \eqref{tetadu*}.
To calculate \eqref{htheta} for $a<b$ we use numerical integration.
The results are depicted in Figure \ref{hbm}.
\begin{figure}[!ht]
\center{
\subfloat{\includegraphics[width=0.45\textwidth]{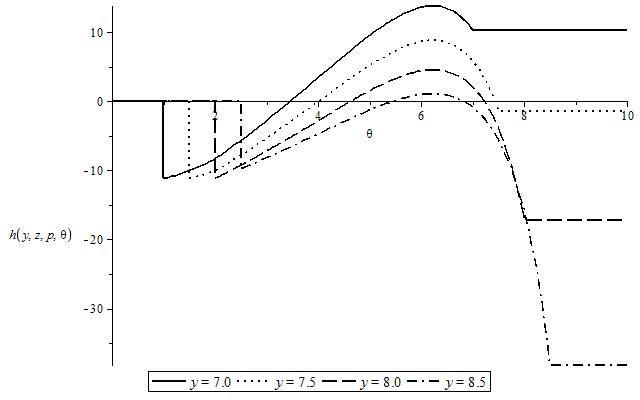}}
\quad
\subfloat{\includegraphics[width=0.45\textwidth]{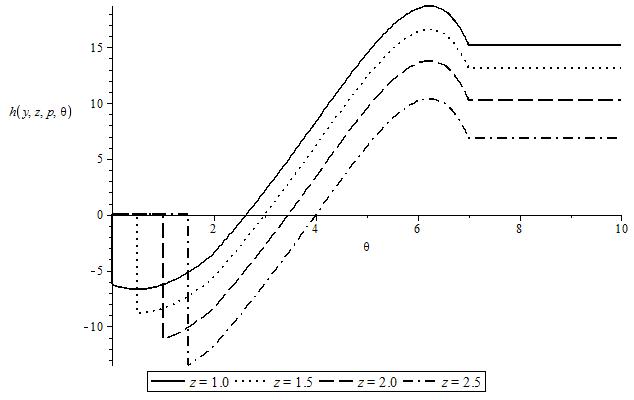}}
\quad
\subfloat{\includegraphics[width=0.45\textwidth]{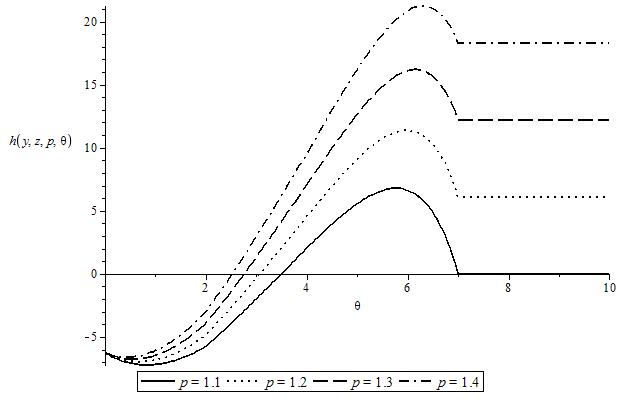}}
\caption{\footnotesize{The value $h(y,z,p,\theta)$ for the drawup contingency contract
for Brownian motion with drift in dependence on $\theta$ for different $y$ (top left),
$z$ (top right) and $p$ (bottom). Parameters: $r=0.01, \mu=0.03, \sigma=0.4, a=10, b=8, \alpha=100, c=50, p=1.35, y=7, z=2$.}}\label{hbm}}
\end{figure}

Note that by \eqref{tetadu*} the optimal drawdown stopping level $\theta^*$  maximizes the function $h$.
From Figure \ref{hbm}, it seems that also for this contract,
there is no dependence of the optimal $\theta^*$ on the initial positions $z$ and $y$ of drawup and drawdown.
However, it is clear that the optimal level of termination
is different for the contracts with and without drawup contingency.
Even if we take the same parameters, the existence of a new parameter,
the starting position of drawup, significantly changes the $\theta^*$ level.

\subsection{Cancellable drawup contingency for $a=b$}
We also analyze results for the special case $a=b$.
To obtain the value of function $h$ we can use \eqref{ha=b}.
Figure \ref{hcl} depicts the results for the Cram\'er-Lundberg model.
We can observe the same lack of dependence of the optimal stopping level $\theta^*$ on the starting position of drawdown and drawup.
\begin{figure}[!ht]
\center{
\subfloat{\includegraphics[width=0.45\textwidth]{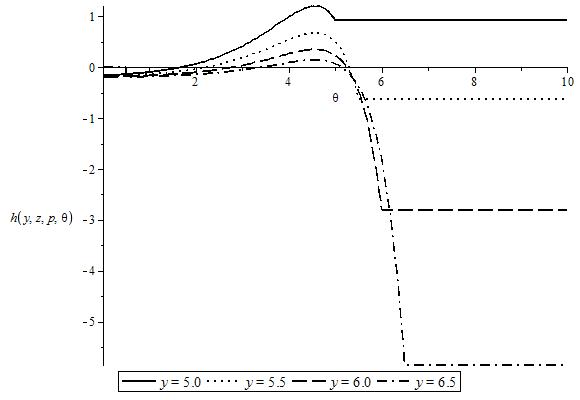}}
\quad
\subfloat{\includegraphics[width=0.45\textwidth]{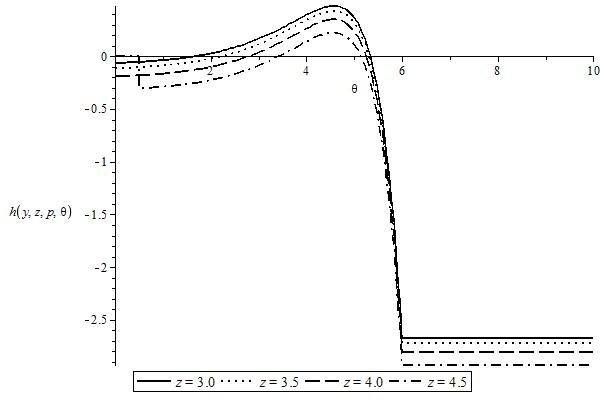}}
\quad
\subfloat{\includegraphics[width=0.45\textwidth]{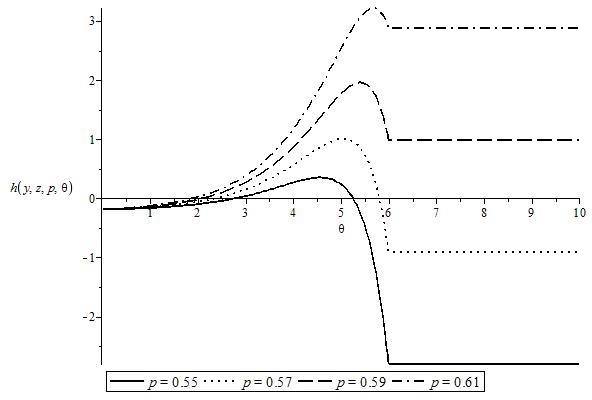}}
\caption{\footnotesize{The value $h(y,z,p,\theta)$ for the drawup contingency contract for the Cram\'er--Lundber in dependence on $\theta$ for different $y$ (top left), $z$ (top right) and $p$ (bottom). Parameters: $r=0.01, \mu=0.04, \beta=0.1, \rho=2.5, a=10, \alpha=100, c=50, p=0.55, y=6, z=4$.}}\label{hcl}}
\end{figure}

\section{Conclusions}\label{sec:con}
In this paper we analyzed a few insurance contracts against drawdown and drawup events of log-return of the asset price.
We model the asset price by a geometric spectrally negative L\'evy process.
We identified the fair premium $p^*$ and the optimal stopping rules for the contracts having a cancellation feature.
We used the theory of optimal stopping and fluctuation theory of L\'evy processes to price these type of contracts.

It is natural to consider other processes than geometric L\'evy processes $S_t=e^{X_t}$,
for example geometric jump-diffusion processes.
It is also of interest to do a more detailed numerical analysis when the jumps are of mixed-exponential type or, more generally, of phase-type.
The idea is to consider all possible shapes of the density of possible jumps in asset prices.
Moreover, one can consider a reward $\alpha$ and a fee $c$ depending on the process $X$ observed at the end of the insurance contract.
In fact, there is still a huge demand for more general insurance contracts that will serve as a policy against major drawdown events.
This will be the subject of future research.


\end{document}